\documentclass[format=acmsmall, review=false]{acmart}
\usepackage{acm-ec-26}

\newif\ifARXIV
\ARXIVtrue   

\ifARXIV
  \AtBeginDocument{%
    \fancypagestyle{firstpagestyle}{%
      \fancyhf{}%
    }%
    \fancypagestyle{acmec}{%
      \fancyhf{}%
      \fancyhead[R]{\thepage}%
    }%
    \renewcommand{\shortauthors}{}%
  }
\fi

\theoremstyle{plain}
\newtheorem{theorem}{Theorem}      
\newtheorem{lemma}{Lemma}          
\newtheorem{proposition}{Proposition}

\theoremstyle{definition}
\newtheorem{definition}{Definition} 

\theoremstyle{remark}

\usepackage{enumitem}
\usepackage[suppress]{color-edits}
\addauthor[David]{dk}{blue}
\addauthor[Evi]{em}{teal}
\addauthor[Xinyu]{xl}{red}

\usepackage[capitalize]{cleveref}
\usepackage{subcaption}
\usepackage{tabularx}

\usepackage{bm}
\usepackage{ifthen}

\usepackage{booktabs} 
\usepackage[ruled,noend]{algorithm2e} 

\SetAlFnt{\small}
\SetAlCapFnt{\small}
\SetAlCapNameFnt{\small}
\SetAlCapHSkip{0pt}
\IncMargin{-\parindent}

\providecommand{\SetCard}[1]{\ensuremath{| #1 |}\xspace}
\providecommand{\SET}[1]{\ensuremath{\left\{ #1 \right\}}\xspace}

\providecommand{\Set}[2]{\ensuremath{\SET{#1 \mid #2}}\xspace}

\providecommand{\Kth}[1]{\ensuremath{{#1}^{\rm th}}}

\providecommand{\PROB}{\ensuremath{\mathbb{P}}\xspace}
\providecommand{\Prob}[2][]{\ensuremath{%
\ifthenelse{\equal{#1}{}}{\PROB[#2]}{\PROB_{#1}[#2]}}\xspace}

\providecommand{\Expect}[2][]{\ensuremath{%
\ifthenelse{\equal{#1}{}}{\mathbb{E}}{\mathbb{E}_{#1}}%
\left[#2\right]}\xspace}

\providecommand{\Event}[2][]{\ensuremath{\ifthenelse{\equal{#1}{}}{%
\mathcal{#2}}{\mathcal{#2}_{{#1}}}}\xspace}

\providecommand{\poly}{\ensuremath{\mathrm{poly}}}

\providecommand{\Interv}[2]{\ensuremath{\left[#1,#2\right]}\xspace}

\providecommand{\vc}[1]{\mathbf{#1}}

\providecommand{\VAL}{\ensuremath{\vc{v}}\xspace}

\NewDocumentCommand{\val}{O{} O{}}{%
\ensuremath{%
  \vc{v}%
  \if\relax\detokenize{#1}\relax\else_{#1}\fi%
  \if\relax\detokenize{#2}\relax\else^{#2}\fi%
}\xspace}

\NewDocumentCommand{\vval}{O{} O{}}{%
\ensuremath{%
  \vc{\widehat v}%
  \if\relax\detokenize{#1}\relax\else_{#1}\fi%
  \if\relax\detokenize{#2}\relax\else^{#2}\fi%
}\xspace}

\newcommand{\pval}{putative valuation\xspace}
\newcommand{\pvals}{putative valuations\xspace}

\NewDocumentCommand{\vtrue}{O{}}{%
  \ensuremath{%
    \vc{v}^{\mathrm{true}}%
    \IfValueT{#1}{\if\relax\detokenize{#1}\relax\else_{#1}\fi}%
  }%
}

\NewDocumentCommand{\bp}{O{}}{%
  \ensuremath{%
    P^{\mathrm{basic}}%
    \if\relax\detokenize{#1}\relax\else_{#1}\fi%
  }%
}

\newcommand{\buffers}{\emph{buffers}\xspace}

\newcommand{\ALG}{\text{ALG}}

\setcitestyle{authoryear}

\title[Learning Fair Allocation of Indivisible Items from \xlreplace{Adversarial}{Limited} Feedback]{Learning Fair Allocation of Indivisible Items from Limited Feedback}

\ifARXIV
  \author{Xinyu Liu}
  \affiliation{\institution{University of Southern California}\city{Los Angeles}\country{USA}}
  \email{xliu4185@usc.edu}
  \author{David Kempe}
  \affiliation{\institution{University of Southern California}\city{Los Angeles}\country{USA}}
  \email{David.M.Kempe@GMail.com}
  \author{Evi Micha}
  \affiliation{\institution{University of Southern California}\city{Los Angeles}\country{USA}}
  \email{emichahy@GMail.com}
\else
  \author{Submission ?} 
\fi

\begin{abstract}
We study a setting in which an algorithm must output a fair allocation of indivisible items while ``learning on the job''.
More specifically, the algorithm is to output an allocation satisfying EF1, PROP1, or similar fairness notions; 
however, the algorithm initially has no information about the agents' valuations, and can only learn about them by (repeatedly) proposing an allocation, and obtaining feedback about a fairness violation in the proposed allocation. 
Importantly, the observed fairness violation may be \emph{adversarially} chosen.
The algorithm's goal is to converge to a fair allocation in a number of rounds polynomial in the number of agents and items, and ideally to also involve only a polynomial amount of computation.

We prove two main results: 
first, when the valuations are additive, then even for mixed items (goods and chores), an allocation satisfying EF1 or PROP1 can be found in polynomial time using the corresponding feedback. 
These results are instantiations of a more general framework which maintains a polytope of candidate valuations consistent with all past feedback. 
The algorithm repeatedly constructs \pvals and uses them to propose allocations; the observed violations then define separating hyperplanes, allowing the algorithm to emulate the ellipsoid method.

When the valuations are monotone, but not necessarily additive, our results are more subtle.
We present an algorithm which is guaranteed to find an EF1 allocation in a number of iterations which is polynomial in the number of agents and items; however, the internal calculations of the algorithm are not guaranteed to be polynomial.
The algorithm again maintains \pvals, and only considers allocations in which each agent obtains the union of an interval plus one additional item with respect to an (arbitrary) ordering of the items.
\dkreplace{We (non-constructively) prove that there always exist EF1 allocations of this form, allowing us to use a further generalization of the preceding ellipsoid-based ideas. 
However}{It is known that such an EF1 allocation always exists; however}, because the existence proof is non-constructive, the internal step of constructing an allocation from the \pvals is not known to take polynomial time.

Finally, we consider weaker and stronger feedback models. 
When the algorithm only learns whether the allocation satisfies the fairness notion, but does not obtain further information on specific violations, we show that any algorithm must in the worst case take a number of iterations exponential in the number of agents; this result is matched by showing that for a constant number of agents, a fair allocation can be found in a polynomial number of iterations.
If the algorithm learns of \emph{all} violations of the proposed allocation, then even for general monotone valuations, a fair allocation can be found in polynomial time, by straightforwardly emulating the Envy-Cycle Elimination Algorithm.
\end{abstract}

\begin{document}


\maketitle

\setcounter{tocdepth}{2} 
\tableofcontents




\section{Introduction}
\label{sec:intro}
Fair division of indivisible goods has numerous  applications, ranging from~allocating first-aid supplies during a crisis or donated food to people in need to distributing compute resources or scheduling access to shared laboratory equipment. 
Over the past two decades, the research community has made significant progress in articulating and formalizing what fairness means in the presence of indivisibilities, as well as how to algorithmically find solutions satisfying these desiderata; see the survey of~\citet{Amanatidis2022}. 
The most classic notion of fairness --- originating in the literature on fair division of \emph{divisible} goods --- is \emph{envy-freeness (EF)}, which requires that no agent prefer another agent’s bundle to their own. 

Exact envy-freeness is often impossible to achieve when items cannot be split\footnote{Simply consider two agents and one desirable item.}, and hence, weaker notions or relaxations of EF have been proposed.
Among these relaxations, perhaps the most frequently studied one is \emph{envy-freeness up to one item (EF1)}, which requires that any envy an agent may have towards another agent can be eliminated by removing a single item from the other agent’s bundle. 
It is known that allocations satisfying EF1 can be computed for broad classes of valuations (e.g., general monotone non-negative valuations) by simple and efficient algorithms such as Envy-Cycle Elimination, and by the even simpler Round Robin algorithm (in which agents take turns selecting their most preferred remaining item) when the valuations are additive~\cite{Lipton2004,Caragiannis2019efx}.

These algorithms are expressed in the ``standard'' algorithmic framework in which the algorithm has full access to the input; here, the input would be numerical utilities or a ranking over all bundles for all players (in the case of general valuations), or the utilities/rankings for individual items (for additive valuations). 
Some work has also considered models in which the algorithm gets to elicit information from players actively, e.g., by eliciting their valuations over an item or a bundle of items~\cite{plaut2020almost, oh2021fairly, li2024complexity} or by asking them to compare pairs of bundles~\cite{BuEtAl2024WINE}. 
Doing so may help decrease the amount of information which agents must communicate about their valuations.

However, in many cases, algorithms must be deployed, and make decisions, without having access to all relevant input. 
In such cases, the algorithm must ``learn from its mistakes'', in the sense of treating the allocation task as an online learning problem. 
As a concrete motivating example, consider the allocation of a resource (such as a lab, robots, or similar equipment) among multiple research groups over the course of a week. Groups may have very idiosyncratic preferences over subsets of allocated time slots, based on contiguity, schedules of specific team members etc. Specifying all of the preferences or utilities ahead of time may be challenging and lead to a combinatorial explosion.  
 
An alternative is for an algorithm to propose an allocation to be used for the following week. 
At the end of the week, groups may provide feedback if they believe that the allocation was not fair --- for instance, if one group received a disproportionate number of desirable slots or if another group was assigned predominantly late-night slots that might be considered less desirable.
These reports could then be taken into account in deciding on the next week's allocation.
The goal for such an algorithm would be to converge quickly to an allocation that satisfies desirable fairness properties. 

This approach recasts the fair allocation problem in the general framework of \emph{online interactive} learning, and more specifically into a model akin to the \emph{Equivalence Query Model} of \citet{angluin:queries-concept}.
Under the Equivalence Query Model, defined by \citet{angluin:queries-concept} for the goal of learning a binary classifer, an online learner repeatedly proposes a candidate binary classifier; unless the classifier happens to be correct, the learner is given a (adversarially chosen) misclassified point.
This model was phrased more abstractly and generalized by \citet{OnlineLearning}.
The abstract form again has the learner proposing a solution, which is either accepted, or results in learning a violated constraint/feature. \dkcomment{Why ``feature''? (I probably wrote it, but am wondering whether it adds anything.)}
As such, the framework captures the challenges of ``learning on the job'', where information can only be obtained in response to an attempted solution, rather than via targeted queries.

In our \emph{Interactive Learning for Fair Division} model, in each round, the learning algorithm proposes a (complete) allocation to the agents, i.e., all items must be allocated.
Unless the allocation satisfies the fairness criterion, the algorithm receives feedback which it can use to propose different allocations in the future.
In keeping with the equivalence query and interactive learning models of \citet{angluin:queries-concept,OnlineLearning}, we assume that the feedback takes the form of \emph{one} violation of a desirable fairness property.
In the prototypical case of EF1 fairness, the violation would be one pair $(i,j)$ of agents such that $i$ envies $j$'s bundle even when any one item is removed. 
If there are multiple such pairs, the pair that is revealed comes with no guarantees; in order to perform a worst-case analysis, we will therefore treat it as adversarial. 

Our model implicitly emphasizes an additional interpretation of EF1.
The standard \emph{normative} intepretation of EF1 (and other fairness notions) is that an allocation is acceptable if it would not give any individual a reason to be too dissatisfied with the allocation.
A subtly different ``practical'' interpretation is that an allocation is acceptable if no individual \emph{does} feel too dissatisfied.
A difference between these two interpretations of course only arises when an allocation is not EF1, but no individual who would be dissatisfied notices (or complains about) the violation.
In particular in the context of online learning, it appears justified to continue using an allocation until a witness of envy-freeness violation is indeed revealed.
Similarly, if no further violations are revealed, the algorithm may even terminate with an allocation that is not EF1. 


Our primary goal is to understand the information-theoretic and computational complexity of finding a fair allocation under the (adversarial) interactive learning model for fair allocations.
That is, how many iterations does an interactive algorithm need until it is guaranteed to find a fair allocation, and what is the time required for the internal computations?

A first natural approach might be to simply give $i$ an item previously allocated to $j$ whenever $i$ envies $j$ even with any one item removed.
In Appendix~\ref{sec:local_adjustment}, we give a simple example showing that even for additive valuations, such an approach may cycle.
Hence, a more sophisticated algorithm is needed.

\subsection{Our contributions}

Our first main result addresses the problem whenever the valuations are \emph{additive}.
In this case, we show that even for mixed goods and chores, an EF1 allocation can be found in polynomially many rounds, with polynomial computation.
The main idea is the following: the algorithm maintains a polytope of \emph{feasible item valuations for all agents}, which are valuations \VAL consistent with all feedback the algorithm has received in the past.
The algorithm then chooses a feasible valuation \VAL as a putative valuation and computes an allocation that would be EF1 under \VAL. 
If the allocation violated EF1, the EF1 violation feedback gives rise to (at least) one hyperplane separating \VAL from the (now updated) polytope of feasible valuations. 
Thus, the feedback on violated EF1 conditions serves as a separation oracle, and choosing the putative valuation according to the ellipsoid algorithm's choice means that the guarantees on the ellipsoid algorithm's convergence transfer to our setting.

This framework can be applied whenever the valuations are captured by polynomially many variables, fairness violations give rise to a separating hyperplane, and there is a polynomial-time algorithm for finding a fair allocation given access to the agents' valuations.
Under these conditions, which include other fairness notions as well as some other settings, our framework yields an interactive learning algorithm for fair allocations in a polynomial number of rounds and using polynomial computation.

In principle, the framework could be applied for general (monotone) valuations; however, the number of variables required to fully encode an agent's valuation function is $2^m$ (where $m$ is the number of items), resulting only in a trivial guarantee on the number of iterations. 
Instead, we consider a different approach, restricting the class of allocations which the algorithm ever proposes.
Specifically, the algorithm fixes an arbitrary ordering over the items, and will only consider allocations in which each agent receives an interval plus at most one additional item; we call such allocations \emph{almost-contiguous allocations}. 
As in the work of \citet{oh2021fairly}, this restriction only serves to reduce the size of the search space, not because intervals are of intrinsic interest. 
We prove non-constructively that there always exists an almost-contiguous EF1 allocation.

After completing this work, we became aware of a result of \citet{Igarashi2023Discrete}, who proves the stronger statement that, for any fixed ordering of the items, there always exists an EF1 allocation in which every agent receives a single interval.  
We had obtained our almost-contiguous EF1 existence result independently by modifying the proof of \citet{Bilo2022Connect} (which showed non-constructively that there always exists an EF2 allocation all of whose bundles are intervals). 
We include our result and proof in Section~5.1 for completeness and present the learning framework in Section~5.2 using this almost-contiguous formulation. 
By using Igarashi’s theorem~\citep[Theorem~3.1]{Igarashi2023Discrete} in place of our existence result, the same learning framework can also be instantiated using only contiguous allocations; we indicate the corresponding consequences at the relevant points.

Our main contribution in this setting is to use this structural restriction to obtain a polynomial bound on the number of interactions. The algorithm still follows the approach of using the ellipsoid algorithm to obtain putative agent valuations. 
The fact that the algorithm only uses bundles which are the union of an interval plus at most one more item implies that it suffices to encode each agent's valuation using polynomially many variables, which in turn yields a polynomial bound on the number of iterations until an EF1 allocation is found. 
However, in each iteration, the algorithm must find an EF1 almost-contiguous allocation with respect to the putative valuations, and the existence of such an allocation was only established non-constructively.\footnote{The proof of \citet[Theorem~3.1]{Igarashi2023Discrete} is likewise non-constructive, and thus does not resolve the computational issue for the contiguous variant.} Indeed, we are not aware of a polynomial-time algorithm for finding such an allocation, implying that we cannot bound the \emph{computation time} by a polynomial.
We also consider a modification of the algorithm, which does not use the Ellipsoid algorithm and feasible/putative valuations at all, but simply avoids allocations containing bundle combinations previously exhibited as EF1 violations. 
We show that while this variant also only requires polynomially many iterations, the computational problem required in the iterations of this modified algorithm is in fact NP-hard.\footnote{The hardness result already holds for the special case of contiguous allocations.}

While our focus is on the interactive learning model, one may naturally wonder how well an algorithm can perform if it receives more or less feedback.

A natural weaker model would be one in which the algorithm only learns that there \emph{exists} some EF1 violation, but no further information; that is, the algorithm obtains only one bit of information about the proposed allocation.
For this model, as we show in Section~\ref{sec:weak}, a number of proposed allocations exponential in the number $n$ of agents is information-theoretically unavoidable, even for the very restrictive special case of identical binary valuations. 
Conversely, by considering a carefully parameterized space of allocations, one can exhaustively search over $m^{O(n)}$ possible allocations, one of which is guaranteed to satisfy EF1; this implies that for a constant number of agents, an EF1 allocation can be found in polynomial time.

At the other extreme is a model in which the algorithm learns \emph{all} pairs $(i,j)$ such that $i$ envies $j$'s bundle even when any one item is removed.\emdelete{\footnote{\dkedit{In line with our ``practical'' interpretation of interactively identifying an EF1 allocation, we believe that requiring a report of \emph{all} EF1 violations places a high burden on the agents, and might result in a brittle algorithm in the presence of agents who, e.g., ignore e-mails about proposed allocations. \dkcomment{Redundant with the ``risky'' comment later?} \xlcomment{I agree that there is some overlap with the "risky" comment in the main text. However, I feel that the footnote makes the issue more concrete.
Given that some reviewers mentioned that they felt our motivation for the single-witness model was not that strong, I am inclined to keep the footnote. Maybe we can move it to after the "risky" sentence in the main text? That way, the reader first sees the point in the main text and then the footnote as a more concrete justification, rather than encountering the footnote first and then seeing a somewhat overlapping statement later.
}}}}
Note that in the equivalence query model of \citet{angluin:queries-concept}, this feedback would correspond to revealing \emph{all} misclassified items, trivially allowing the learner to infer the correct classifier after one round.
While the situation is not quite as trivial for the goal of finding fair allocations, we show in Section~\ref{sec:all-witness} that it is fairly easy to simulate the Envy-Cycle Elimination Algorithm with only polynomial blowup in the number of rounds. Thus, one can, in a polynomial number of rounds and with polynomial computation, find an EF1 allocation even when the agents' valuations are arbitrary monotone functions.
This is further evidence that this model gives the algorithm too much information; from a practical perspective, relying on the fact that all agents reveal all of their envy is clearly also risky.
We briefly remark that if the algorithm instead of \emph{all} EF1 violation leans only a single i.i.d.~uniformly random EF1 violation, then by proposing the same allocation multiple times, in a number of rounds polynomial in $n$, the algorithm can with high probability obtain all EF1 violations after all, so this model only results in polynomial blowup compared to the preceding one. 

\subsection{Related Work}
\label{sec:related-work}
Our work contributes to the growing literature on fair division of indivisible goods when the preferences of the agents are not known. Multiple works have focused on value queries, in which the algorithm may ask an agent for the value of a specific bundle~\citep{plaut2020almost,oh2021fairly}, or   partial-ranking feedback, where agents provide coarse ordinal information over items or bundles~\citep{halpern2021fair,benade2022dynamic}. \citet{feige2025low} recently studied the communication complexity of fair division and quantified how many bits an agent must report in order for a fair allocation to be computed. Other recent works have also considered models based on comparison queries, in which agents compare different items or bundles without revealing exact numerical values~\citep{BuEtAl2024WINE,li2024complexity}. In contrast, our model does not allow targeted queries of any kind. Each interaction consists of proposing a complete allocation, and the only feedback the algorithm receives --- when the allocation is not fair --- is a witness of violation, such as an EF1 envy pair. Thus, the feedback is substantially more limited than in standard preference-query models: the algorithm cannot isolate a single pair of bundles to compare, cannot request information about a specific agent, and cannot probe the valuation of any sub-bundle. All information must be inferred from counterexamples produced in response to full allocations, placing our framework much closer to mistake-driven and interactive online learning models than to traditional elicitation.

Interactive learning has been studied in multiple settings. 
To the best of our knowledge, the earliest work formalizing the idea of trial-and-error learning through local mistakes is the seminal work of~\citet{angluin:queries-concept}, which introduced the \emph{Equivalence Query Model}. In this model, a learner proposes a binary classifier, and if the classifier is incorrect, the learner is given an adversarially chosen misclassified point. 
As explained by \citet{angluin:colt-survey}, this model is equivalent to the simultaneously proposed and more widely known standard Online Learning model of \citet{littlestone:learning}, in which the learner is given one example to classify in each round, and learns whether this example was classified correctly, using the feedback to update its (internal) classifier.

The equivalence query framework was generalized by~\citet{OnlineLearning}, who studied interactive problems represented as weighted undirected graphs $(V,E)$: the target solution is a node $s^*$, and whenever the learner proposes a node $s$, the feedback consists of a neighboring node $s'$ that is strictly closer to $s^*$. Prior work by \citet{BinarySearch} had shown that a learner can identify $s^*$ within $O(\log |V|)$ interactions, directly implying efficient interactive learning algorithms for classification, ranking, and clustering.
While the graph-based framework is a powerful abstraction, our interactive fair-division model does not fall into this framework. 
When an allocation violates EF1, the learner receives only a witness of violation --- for example, a pair $(i,j)$ such that agent~$i$ envies agent~$j$. 
This feedback does not necessarily guide the learner to an allocation that is ``closer'' to a valid EF1 solution in any graph-theoretic or metric sense. 

The interactive learning model has been studied across multiple settings. In clustering, a learner proposes candidate clusterings and receives limited feedback, often in the form of requests to merge or split clusters~\cite{AdaptiveHierarchical, balcan:blum:split-merge, awasthi:zadeh:supervised-clustering, awasthi:balcan:voevodski:local-algorithm-journal}. 
In stable matching, interactive feedback identifies blocking pairs; information on these blocking pairs guides the learner toward a stable matching~\cite{LearningStableMatching, bei2013complexity}. 
In ranking, coactive and interactive query models learn a target ranking by allowing the learner to propose an option and having the user return a preferred alternative higher in their ranking~\cite{shivaswamy2015coactive}. Recent work has also explored interactive learning for computing voting rules, where the learner iteratively proposes outcomes, and agents respond with improvements according to their rankings~\cite{michacomputing2025}. 
Finally, in training large language models, interactive and coactive learning frameworks enable models to refine their outputs through localized human-feedback improvements~\cite{tucker2024coactive}. 
To the best of our knowledge, we are the first to introduce an interactive learning framework for fair division.

Somewhat similar ideas to the ones we use are used in the context of \emph{multidimensional binary search} and its generalizations \citep{paes-leme:schneider:contextual-search,cohen:lobel:paes-leme:dynamic-pricing,lobel:paes-leme:vladu:multidimensional-binary-search,KLPS:contextual-search-adversarial}.
This model is motivated by online problems such as contextual pricing.
In the most basic model, there is a product with $d$ relevant features, and each user's utility (or willingness to pay) is obtained as a linear combination of these features.
The feature vector $\vc{\theta}$ is unknown, and can/must be inferred by interacting with a sequence of agents. 
Each agent/context $k$ is characterized by a vector $\vc{x}_k$, and the algorithm can ask a simple query about the inner product $\vc{\theta} \cdot \vc{x}_k$, such as asking to compare it to a proposed threshold (which would correspond to observing whether the agent bought the product at the proposed price).
Such queries can be construed as revealing separating hyperplanes for the vector $\vc{\theta}$, as in our algorithms.
As in our case, the core of the analysis is to show that the remaining feasible region for $\vc{\theta}$ shrinks sufficiently fast in a meaningful way.
A key difference is in the level of control. 
In the work mentioned above, the algorithm can choose the threshold given the context, whereas in our work, the algorithm proposes an allocation, but does not know about the normal of the hyperplane it will observe.







\section{Preliminaries and Model Setup}
\label{sec:prelim}
Let $N = [n]$ be the set of agents and $M$ the set of $m$ indivisible items. 
Each agent $i \in N$ has a valuation function $v_i : 2^{M} \to \mathbb{R}$ assigning a real value to every bundle of items.  
Without loss of generality, we assume that the valuations satisfy \(v_i(\emptyset) = 0\).

We consider three standard settings:
\begin{itemize}
  \item \textbf{Goods setting:} \(v_i(S) \ge 0\) for all \(S \subseteq M\), 
       capturing allocations of desirable items such as resources or prizes.
  \item \textbf{Chores setting:} \(v_i(S) \le 0\) for all \(S \subseteq M\), 
        corresponding to allocations of undesirable tasks or chores.
  \item \textbf{Mixed setting:} valuations may take both positive and negative values, 
        capturing combinations of goods and chores.
\end{itemize}

{A valuation function is \emph{additive} if for every bundle $S \subseteq M$, $ v_i(S) = \sum_{x \in S} v_i(x)$, where $v_i(x) = v_i(\SET{x})$ denotes the value of item $x$ to agent~$i$. 
A more general class is the class of all monotone valuations.
A valuation $v_i$ (for goods) is \emph{monotone} if for all $S,T \subseteq M$ with $S \subseteq T$, we have $v_i(S) \le v_i(T)$. 
An \emph{allocation} $A = (A_1,\ldots, A_n)$ assigns each agent~$i \in N$ a bundle $A_i \subseteq M$ of items such that the bundles form a partition of $M$, i.e.,
$A_i \cap A_j = \emptyset$  for  $i \ne j$,  and $
\bigcup_{i=1}^n A_i = M$.}

\subsection{Fairness Notions}

We are specifically interested in finding \emph{fair} allocations, and we now define suitable notions of fairness.
Classical notions of fairness are defined in terms of agents’ valuations over their assigned bundles.  
An allocation is \emph{envy-free (EF)} if no agent prefers another’s bundle to their own, i.e., $v_i(A_i) \ge v_i(A_j)$ for all $i,j \in N$.  
Another standard notion is \emph{proportionality (PROP)}, which requires each agent to receive at least a $1/n$ fraction of their total value for all items, i.e., $v_i(A_i) \ge v_i(M)/n$.  
However, allocations satisfying EF or PROP do not always exist for indivisible items, as can be seen from the simple case of two agents and one item.

To address this issue, \citet{Budish2011} introduced \emph{envy-freeness up to one good (EF1)} as a relaxation of EF, and \citet{Conitzer2017} proposed the analogous notion of \emph{proportionality up to one good (PROP1)}. 

More recently, \citet{Aziz2022ef1mixed} extended both notions to settings involving mixtures of goods and chores, 
showing that EF1 and PROP1 allocations always exist for additive valuations.  

\begin{definition}[EF1 for mixed goods and chores]
We say that agent $i$ \emph{EF1-envies} agent $j$, and write $A_i \prec_i^{\mathrm{EF1}} A_j$, if both of the following hold:
\begin{align*}
v_i(A_i) & < v_i(A_j) &
v_i(A_i \setminus \SET{x}) & < v_i(A_j \setminus \SET{x})
           \text{ for all } x \in A_i \cup A_j.
\end{align*}

An allocation $A$ is \emph{EF1} if there is no pair $i,j$ with $A_i \prec_i^{\mathrm{EF1}} A_j$.
\end{definition}

The first condition captures that $i$ envies $j$'s bundle.
The second condition states that there is no item $x$ to remove from $A_j$ (in the case of goods) that would make $i$ not envy $j$ any more, 
and there is no item $x$ to remove from $A_i$ (in the case of chores) that would make $i$ not envy $j$ any more.
Together, the conditions capture that no single item can be removed from either bundle to make $i$ not envy $j$ any more.

\begin{definition}[PROP1 for additive mixed goods and chores]
We say that $i$ is PROP1-envious, and write $A_i \prec_i^{\mathrm{PROP1}} M$, if all of the following hold:
\begin{align*}
v_i(A_i) & < \tfrac{1}{n} \cdot v_i(M) \\
v_i(A_i) + v_i(x) & < \tfrac{1}{n} \cdot v_i(M)
            \text{ for all } x \in M \setminus A_i \\
v_i(A_i) - v_i(x) & < \tfrac{1}{n} \cdot v_i(M)
            \text{ for all } x \in A_i.
\end{align*}

An allocation $A$ satisfies \emph{PROP1} if there is no agent $i$ with $A_i \prec_i^{\mathrm{PROP1}} M$.
\end{definition}

This definition mirrors the EF1 definition for the mixed setting.
The second and third conditions capture that proportionality cannot be achieved by adding at most one good to agent~$i$’s bundle or removing at most one chore from it, respectively.

\citet{Conitzer2017} showed that, under additive valuations in the goods setting, every EF1 allocation is also PROP1.  
For general monotone valuations over goods, however, this implication does not hold.  
The well-known Envy-Cycle Elimination algorithm of \citet{Lipton2004} guarantees the existence of EF1 allocations for monotone goods, but allocations satisfying PROP1 may not exist even in very simple instances.  
Consider two agents and three items, with both agents valuing the entire set at~1 and every strict subset at~0; any allocation violates PROP1.  
Consequently, most studies of PROP1 focus on additive valuations, while EF1 has been explored more broadly.

\citet{Aziz2022ef1mixed} further generalized the implication of \citet{Conitzer2017} from goods to mixed settings, 
proving that every EF1 allocation is also PROP1 under additive valuations.

\begin{proposition}[{\citet[Proposition~2]{Aziz2022ef1mixed}}]
\label{prop:ef1-implies-prop1}
In the mixed setting with additive valuations, every EF1 allocation is also PROP1.
\end{proposition}

\subsection{Interactive Feedback Models}
\label{subsec:models}

We study the problem of finding fair allocations via interactive algorithms that ``learn on the job''.
Instead of having access to agents’ valuation functions or direct queries,
the algorithm can only interact with the environment by proposing an allocation and observing the resulting feedback.
Our main focus is on the fairness notions EF1 and PROP1.
In each round $t = 1,2,\ldots$, the algorithm proposes an allocation
$A^{(t)} = (A^{(t)}_1,\ldots,A^{(t)}_n)$.
If the proposed allocation satisfies the target fairness notion, the interaction terminates successfully.
Otherwise, the algorithm receives feedback indicating that the allocation violates the fairness requirement.

Different feedback models correspond to different levels of information about fairness violations.
In particular, when a proposed allocation is not fair, the feedback may reveal information about the specific way in which fairness was violated.
At one extreme, the feedback may only indicate whether the allocation satisfies the target fairness notion, without revealing any additional information.
At the other extreme, the feedback may reveal all individual violations, e.g., all pairs of agents $i, j$ such that $i$ envies $j$ even after the removal of any one item from $j$'s bundle (EF1) or all agents $i$ such that $i$ does not receive a proportional share (PROP1).
Our focus lies on intermediate models in which the feedback reveals a single violation.
To ensure robustness of our algorithms, we assume that whenever such a witness is revealed, it is chosen \emph{adversarially}.
We now formalize such intermediate feedback models for EF1 and PROP1.

\begin{definition}[Single-Witness EF1 Model]
\label{def:single-witness-ef1}
In the Single-Witness EF1 Model, upon proposing an allocation $A$,
the feedback either \emph{certifies} that $A$ is EF1,
or consists of an \emph{adversarially chosen violating pair} $(i,j) \in N \times N$ with $A_i \prec_i^{\mathrm{EF1}} A_j$.
We call the pair $(i,j)$ a \emph{witness} of EF1 violation. 
\end{definition}

\begin{definition}[Single-Witness PROP1 Model]
\label{def:single-witness-prop1}
In the Single-Witness PROP1 Model, upon proposing an allocation $A$,
the feedback either \emph{certifies} that $A$ is PROP1,
or consists of an \emph{adversarially chosen violating agent} $i \in N$ with $A_i \prec_i^{\mathrm{PROP1}} M$.
We call $i$ a \emph{witness of PROP1 violation}.
\end{definition}

For completeness, we also consider the two extreme feedback models for EF1, corresponding to revealing no witness or revealing all violating witnesses.

\begin{definition}[Membership EF1 Model]
\label{def:membership-ef1}
In the \emph{Membership EF1} Model, upon proposing an allocation $A$,
the feedback consists of a single bit indicating whether $A$ is EF1.
This represents the minimal feedback setting.
\end{definition}

\begin{definition}[All-Witness EF1 Model]
\label{def:all-witness-ef1}
In the \emph{All-Witness EF1 Model}, upon proposing an allocation $A$,
the feedback either \emph{certifies} that $A$ is EF1,
or returns the set of all pairs $(i,j) \in N \times N$ witnessing EF1 violation.
\end{definition}

It is worth noting that a model in which a single violating witness is chosen uniformly at random in each round
is comparable to the All-Witness Model.
Since there are only polynomially many potential violating pairs for EF1,
standard coupon collector arguments imply that repeating the interaction polynomially many times will reveal all violations with high probability.
For this reason, and to ensure robustness to worst-case interaction,
we model the single-witness feedback as being chosen \emph{adversarially}.

Having defined these feedback models, we now turn to the algorithmic questions they give rise to.
We study the interaction and computational complexity of finding fair allocations under the above interactive feedback models.

Throughout the paper, each interaction consists of the algorithm proposing an allocation and receiving the corresponding feedback.
We measure the algorithm’s \emph{interaction complexity} by the total number of interactions, equivalently, the total number of allocation proposals it makes.
An algorithm is said to use a polynomial number of interactions if this number is bounded by $\poly(n,m)$.
Separately, we measure the algorithm’s \emph{computational complexity} in the standard sense, i.e., the total running time required by the algorithm.

For convenience, we refer to interactions under the Single-Witness EF1, Single-Witness PROP1, Membership EF1, and All-Witness EF1 Models as \emph{single-witness EF1}, \emph{single-witness PROP1}, \emph{membership EF1}, and \emph{all-witness EF1} interactions, respectively.
We will use the terms ``interaction'' and ``allocation proposal'' interchangeably to denote a single interaction between the algorithm and the feedback model.
An allocation that is certified by the feedback as satisfying EF1 or PROP1 will be called a \emph{certified EF1} or \emph{certified PROP1} allocation.

\section{Polynomial-Time EF1 in the Single-Witness Model under Additive Valuations}
\label{sec:ef1_additive}
We start by presenting a polynomial-time algorithm that  finds an EF1 allocation under \emph{additive valuations} in the mixed-items setting with polynomially many single-witness EF1 interactions. 
The key idea is to have the algorithm guess a \pval for each agent that is not in conflict with any feedback received so far.
Each witness EF1 feedback can then be interpreted as imposing hyperplane constraints on the corresponding agent's valuation.
These constraints are used to iteratively refine the set of feasible \pvals through an ellipsoid-based update process until an EF1 allocation is certified. 

In Section~\ref{sec:generalizing-ellipsoid}, we will show how this framework extends to other fairness notions, such as proportionality up to one item (PROP1).

\subsection{Overview of the Algorithm}

Our algorithm alternates between computing candidate allocations under the current
guess of the agents’ valuations and refining these guesses using feedback received.  
At a high level, the algorithm maintains for each agent~$i$ a~\emph{\pval}  vector~$\vval[i]$, representing the algorithm’s current estimate of $i$’s (unknown) valuation.

In each iteration, the algorithm calls a full-information EF1 subroutine as a black box,
which returns an EF1 allocation under the current~$\vval$.
In our implementation (Algorithm~\ref{alg:ellipsoid-ef1}), we instantiate the subroutine with the \textsc{DoubleRoundRobin} algorithm of \citet[Algorithm~1]{Aziz2022ef1mixed}, which runs in polynomial time and guarantees EF1 under additive valuations even for mixed valuations. However, in principle, any full-information EF1 algorithm suitable for the class of valuation functions under consideration can be used.

The allocation computed by the EF1 subroutine is then proposed, and feedback is received in response.
If the feedback certifies that the allocation is EF1 under the agents' (true) valuations, the algorithm terminates.
Otherwise, the feedback identifies a violating pair $(i,j)$, and we will show that such a pair always yields a linear constraint separating $i$'s true valuation from $i$'s current \pval.
This constraint is then incorporated into the feasible valuation region, and the algorithm uses an ellipsoid-based update to refine its \pvals, 
excluding the current (incorrect) estimate and moving closer to the true valuation profile.

This iterative process repeats until an allocation is certified EF1.  
The key structural insight is that each round either identifies an EF1 allocation or rules out a portion of the valuation space via a separating constraint. 
As a result, the algorithm can exactly be viewed as an execution of the ellipsoid algorithm for Linear Programming with a separation oracle, which is known to terminate in a number of iterations polynomial in the number of variables; the number of variables here equals the number of valuations $\widehat{v}_i(x)$, i.e., it is $nm$.

\begin{algorithm}[ht]
\caption{\textsc{Ellipsoid-EF1} (Single Witness Feedback, Mixed Setting, Additive Valuations)}
\label{alg:ellipsoid-ef1}
\DontPrintSemicolon
\SetKwInput{KwInput}{Input}
\SetKwInput{KwOutput}{Output}
\KwInput{Agents $N$ and items $M$.}
\KwOutput{An EF1 allocation with respect to the (unknown) true additive valuations.}
\BlankLine

\lForAll{$i\in N$}{initialize (normalized) feasible region $\mathcal{R}_i^{(0)}\!\gets\![-1,1]^m$.}
$t \gets 0$.\;
\While{no allocation has been certified as EF1}{
  $t\!\gets\!t\!+\!1$.\;

  \tcp{\textbf{(a) Valuation Update and Allocation step}  }
  $\vval[k][(t)] \gets \mathrm{Center}\left(\mathcal{R}_k^{(t-1)}\right)$ for all $k \in [n]$.\;
  $A^{(t)} \gets \textsc{DoubleRoundRobin}(\vval[][(t)])$.
  \tcp*[r]{Run DoubleRoundRobin on current valuations}
  Propose the allocation $A^{(t)}$.\;
  \lIf{$A^{(t)}$ is certified EF1}{\Return $A^{(t)}$.}
  \BlankLine
  \tcp{\textbf{(b) Identify new separating constraints}}
  Receive a violating pair $(i,j)$.\;
    
  Create new constraints
  \(
      \mathcal{C}_i^{(t)} :=
      \begin{aligned}[t]
        &\Set{\val[i]}{v_i(A_i^{(t)}) \le v_i(A_j^{(t)})} \\
        &\cap \bigcap_{x \in A_i^{(t)}}
            \Set{\val[i]}{v_i(A_i^{(t)} \setminus \SET{x}) \le v_i(A_j^{(t)})} \\
        &\cap \bigcap_{y \in A_j^{(t)}}
            \Set{\val[i]}{v_i(A_i^{(t)}) \le v_i(A_j^{(t)} \setminus \SET{y})} .
      \end{aligned}
  \)
  
  \BlankLine
  \tcp{\textbf{(c) Ellipsoid update}}
  Update feasible regions:\;
  $\mathcal{R}_i^{(t)} \gets \mathcal{R}_i^{(t-1)} \cap \mathcal{C}_i^{(t)}$.\; 
  $\mathcal{R}_k^{(t)} \gets \mathcal{R}_k^{(t-1)}$ for all $k\neq i$.\;
  }
\end{algorithm}

\begin{theorem}
\label{thm:ef1-ellipsoid-poly}
Under additive valuations in the mixed setting,  
Algorithm~\ref{alg:ellipsoid-ef1} always terminates with an EF1 allocation within $\poly(n,m)$ single-witness EF1 interactions, and runs in total time $\poly(n,m)$.
\end{theorem}

Before proving Theorem~\ref{thm:ef1-ellipsoid-poly},
we recall the ellipsoid method~\citep{Khachiyan1979,GLS1981} and explain how the separating constraints derived from the single-witness feedback are used to refine the feasible valuation region and update the \pvals in Algorithm~\ref{alg:ellipsoid-ef1}.

\subsection{Ellipsoid-Based Valuation Refinement}
\label{sec:ellipsoid}

The ellipsoid method~\citep{Khachiyan1979,GLS1981} provides a general polynomial-time procedure  
for solving linear feasibility problems of the form 
\[
\text{find } \vc{x} \in \mathbb{R}^d \text{ such that } A \vc{x} \le \vc{b},
\]
given access to a \emph{separation oracle}. 
For any candidate point~$\vc{\widehat{x}}$, the oracle either confirms that $\vc{\widehat{x}}$ is feasible,
or returns a violated inequality ${\vc{a}^{\top}_i} \vc{\widehat{x}} > {b_i}$,
where $\vc{a}_i^\top$ is the \Kth{i} row of the matrix $A$ and \(b_i\) the corresponding entry of the vector \(\vc{b}\),
thereby separating~$\vc{\widehat{x}}$ from the feasible region. 

The ellipsoid algorithm maintains an ellipsoid~$E_t$ that always contains the current feasible region.  
In each iteration, it queries the oracle at the center of~$E_t$:  
if the oracle certifies feasibility, the algorithm terminates;  
otherwise, the returned separating hyperplane cuts off the current center while preserving all feasible points.  
The next ellipsoid \(E_{t+1}\) is chosen to contain the intersection of \(E_t\) with the feasible halfspace identified by the oracle.
Each iteration shrinks the ellipsoid’s volume by a fixed multiplicative factor (depending only on the dimension~\(d\)), while ensuring that all feasible points remain inside.
After a number of oracle calls polynomial in the dimension~$d$,  
the method either finds a feasible point (if one exists) or correctly certifies that no feasible solution exists.

\smallskip
\paragraph{Ellipsoid updates in our algorithm.}
In Algorithm~\ref{alg:ellipsoid-ef1}, the ellipsoid method is applied separately to the
valuation space of each agent.  For a fixed agent $i$, the algorithm maintains a convex
feasible region $\mathcal{R}_i^{(t)} \subseteq \mathbb{R}^m$ consisting of all additive
valuations not yet ruled out by previously derived EF1-based constraints.  The \pval $\vval[i][(t+1)]$ used by the algorithm is the center of the ellipsoid enclosing
$\mathcal{R}_i^{(t)}$.

When an allocation constructed from all of the $\vval[i][t]$ is proposed, it is either certified as EF1, or a violating pair $(i,j)$ is returned.
Certification corresponds to the separation oracle declaring the current ellipsoid center feasible. 
A violating pair corresponds to the infeasible case in the ellipsoid framework. 
In this case, each EF violation implied by the EF1 violation gives rise to a linear constraint. 
We will show in Section~\ref{sec:correctness} that among these constraints, at least one must be \emph{new} and provides the separating hyperplane required by the ellipsoid method.

After the constraint is obtained, the feasible region is refined according to
\[
\mathcal{R}_i^{(t)}
=
\mathcal{R}_i^{(t-1)} \cap \mathcal{C}_i^{(t)},
\]
where $\mathcal{C}_i^{(t)}$ denotes the polyhedral region defined by the intersection of the halfspaces corresponding to all linear constraints implied by the feedback received at iteration~$t$.
Since the algorithm incorporates all such constraints, including additional (possibly redundant) ones that do not affect correctness, 
the feasible region can only shrink further.
The ellipsoid method then computes a new ellipsoid enclosing $\mathcal{R}_i^{(t)}$, and its center becomes the next \pval $\vval[i][(t)]$. 
Consequently, each iteration further refines the valuation estimates.

\subsection{Correctness and Complexity}\label{sec:correctness}

\begin{proof}[Proof of Theorem~\ref{thm:ef1-ellipsoid-poly}]
Having established the ellipsoid framework as the information-update mechanism,
we now show that Algorithm~\ref{alg:ellipsoid-ef1} always terminates with an EF1 allocation,
making only $\poly(n,m)$ allocation proposals and running in total time $\poly(n,m)$.

\paragraph{Correctness.}
We view the space of each agent’s additive valuations as a convex region in~$\mathbb{R}^m$,
where each coordinate corresponds to the value assigned to one item.
The restriction to valuations in $[-1,1]$ is without loss of generality, as all valuations can be scaled down by the maximum individual valuation for any item. Hence, we assume that the (unknown, adversarial) ground truth valuations all lie in $[-1,1]$.
The algorithm does not attempt to reconstruct the true valuation~$\vtrue[i]$, but to iteratively refine a \pval~$\vval[i][(t)]$ until the allocation produced under~$\vval[][(t)]$ is certified EF1 by the adversary.

By the correctness of the \textsc{DoubleRoundRobin} subroutine~\citep{Aziz2022ef1mixed}, if the feedback is a violating pair~$(i,j)$,
then the current \pval~$\vval[t][(t)]$ cannot coincide with~$\vtrue[i]$:
after all, under $\vtrue[i]$, \textsc{DoubleRoundRobin} would not produce a violation for~$i$.
Hence, the appearance of~$i$ as the first index in a violating pair
provides concrete evidence that its current \pval~$\vval[i][(t)]$ is inconsistent with the true valuations.

In step~(b) of Algorithm~\ref{alg:ellipsoid-ef1}, the feedback $(i,j)$ is interpreted as indicating that, under the valuations $\vtrue[i]$, 
$i$ continues to envy agent~$j$ regardless of whether no item is removed or which single item is removed from either the bundle $A_i^{(t)}$ or $A_j^{(t)}$.

Accordingly, the algorithm adds the corresponding linear constraints:
the constraint $v_i(A_i^{(t)}) \le v_i(A_j^{(t)})$, as well as, for every $x \in A_i^{(t)}$, the constraint $v_i(A_i^{(t)} \setminus \SET{x}) \le v_i(A_j^{(t)})$, 
and for every $y \in A_j^{(t)}$, the constraint $v_i(A_i^{(t)}) \le v_i(A_j^{(t)} \setminus \SET{y})$.
All these constraints must be satisfied by $\vtrue[i]$, 
and therefore define halfspaces that are added to refine the feasible region in step~(b).

On the other hand, since the allocation constructed from $\vval[][(t)]$ is EF1 under $\vval[i][(t)]$, there exists at least one case 
(either no item removal or the removal of a single item from $A_i^{(t)}$ or $A_j^{(t)}$) 
under which agent~$i$ would no longer envy agent~$j$ if the valuations were $\vval[i][(t)]$.
The corresponding inequality is therefore violated by $\vval[i][(t)]$ but satisfied by $\vtrue[i]$, 
and hence defines a valid separating hyperplane for the ellipsoid update.

{
Since the feasible region is initialized as the entire (normalized) valuation space and every added halfspace contains $\vtrue[i]$, 
the region $\mathcal{R}_i^{(t)}$ is nonempty throughout the execution of the algorithm.
Furthermore, whenever the current allocation is not certified EF1, the received feedback identifies a pair $(i,j)$ that induces a collection of linear constraints, 
among which at least one excludes the current estimate $\vval[i][(t)]$, 
while all constraints are satisfied by all valuations consistent with the feedback (including $\vtrue[i]$).}
Consequently, each non-terminating round introduces a valid separating constraint and therefore makes progress by the ellipsoid method.

By the convergence guarantee of the ellipsoid method, after at most polynomially many (see below) iterations, the ellipsoid method must terminate. At that point, there cannot be any more separating hyperplanes, meaning that the allocation computed using the \pvals $\vval$ must be EF1 under the (unknown) true valuations $\vtrue$. 
Note that this does not necessarily imply that $\vval = \vtrue$; the algorithm will typically terminate much earlier.

\paragraph{Complexity.}
Each iteration performs three operations.
\begin{enumerate}
\item \emph{Allocation step.} 
Each run of \textsc{DoubleRoundRobin} takes $O(\max (m^2, mn))$ time 
(by \citet[Theorem~1]{Aziz2022ef1mixed}).

\item \emph{Witness processing.}
Given the violating pair $(i,j)$, the algorithm constructs the full collection of linear constraints corresponding to the EF1 conditions for the agent pair $i,j$, including the case of no item removal as well as all single-item removals from $A_i^{(t)}$ and $A_j^{(t)}$.
Evaluating the three EF1 conditions takes
$O(1 + |A_i^{(t)}| + |A_j^{(t)}|) = O(m)$ time.   
Each resulting linear inequality involves at most $m$ item variables, and at most
$O(m)$ such inequalities are generated, so writing down all corresponding linear
constraints requires $O(m^2)$ time.

\item \emph{Ellipsoid update.} The algorithm maintains an ellipsoid in $\mathbb{R}^m$ for each agent $i$.
Each added halfspace triggers a standard ellipsoid update, which requires $\poly(m)$ arithmetic time.  
In each iteration, at most $O(m)$ such halfspaces are generated for agent $i$, so the total ellipsoid-update time per iteration is still $\poly(m)$.
\end{enumerate}

By the weakly polynomial bounds of the ellipsoid method
(e.g., \citealp{Khachiyan1979,GLS1981}), the number of separation calls (and thus
number of proposed allocations) needed per agent is $\poly(m, L)$, where $L$ bounds
the bit-length of numbers used in the constraints.\footnote{In our setting, the
separating inequalities are sums of item values and have coefficients in $\{0,1\}$; thus $L$ can be taken polynomial in the input description of item values.} 
Since in each non-terminating round, exactly one agent~$i$ is updated,
the total number of rounds is $\sum_{i=1}^n \poly(m,L)=\poly(n,m)$, 
and the total arithmetic time is also $\poly(n,m)$,
accounting for the call to \textsc{DoubleRoundRobin}, witness identification,
and the ellipsoid update in each iteration.

As the existence of the true valuation $\vtrue[i]$ ensures that the feasible region for each agent remains nonempty throughout the execution of the algorithm, 
and \textsc{DoubleRoundRobin} is guaranteed to output an EF1 allocation under the true additive mixed valuations \citep[Theorem~1]{Aziz2022ef1mixed}, 
the ellipsoid method must eventually terminate, and at that point, the allocation is indeed EF1 with respect to the true valuations.
Consequently, the overall procedure must terminate within $\poly(n,m)$ iterations and total $\poly(n,m)$ time with an allocation certified EF1.
\end{proof}

\section{Generalizing the Ellipsoid Framework}
\label{sec:generalizing-ellipsoid}
We now show that the ellipsoid-based approach introduced in Section~\ref{sec:ef1_additive} extends beyond envy-freeness and yields a general framework for computing fair allocations from limited interactive feedback.
The framework applies whenever a fairness violation, together with the available feedback,
can be translated into a separating constraint on the agents' valuations.

We first illustrate the framework by instantiating it for \emph{proportionality up to one item (PROP1)} under additive valuations.
This example shows that the same algorithmic structure and correctness argument extend beyond EF1.
We then abstract the common structure underlying the EF1 and PROP1 instantiations to formulate a general ellipsoid-based fairness framework.

\subsection{An Instantiation: PROP1 under Additive Valuations}
\label{subsec:prop1-instantiation}

To adapt the framework to PROP1, three components may in principle require modification:
the allocation subroutine, the form of witness-based feedback, and the rule for identifying new separating constraints in step~(b).

However, under additive valuations, EF1 implies PROP1, and hence any subroutine
that finds an EF1 allocation under the true valuations also guarantees PROP1.
We therefore retain the same \textsc{DoubleRoundRobin} algorithm as the allocation module,
and only replace the EF1 witness with its PROP1 counterpart,
while updating the constraint construction rule accordingly.
The resulting procedure, given as Algorithm~\ref{alg:ellipsoid-prop1}, differs from
Algorithm~\ref{alg:ellipsoid-ef1} only in the type of witness feedback and corresponding constraints.
Its correctness and polynomial complexity follow by arguments analogous to the EF1 case.

\begin{algorithm}[t]
\caption{\textsc{Ellipsoid-PROP1} (Single Witness Feedback, Mixed Setting, Additive Valuations)}
\label{alg:ellipsoid-prop1}
\DontPrintSemicolon
\SetKwInput{KwInput}{Input}
\SetKwInput{KwOutput}{Output}
\KwInput{Agents $N$ and items $M$\dkedit{.}}
\KwOutput{An EF1 allocation with respect to the (unknown) true additive valuations.}
\BlankLine

\lForEach{$i\!\in\!N$}{
  Initialize (normalized) feasible region $\mathcal{R}_i^{(0)}\!\gets\![-1,1]^m$.}
$t \gets 0$.\;

\While{no allocation has been certified as PROP1}{
  $t \gets t+1$.\;

  \tcp{\textbf{(a) Valuation update and allocation step}}
  $\vval[k][(t)] \gets \mathrm{Center}\big(\mathcal{R}_k^{(t-1)}\big)$ for all $k \in N$.\;
  $A^{(t)} \gets \textsc{DoubleRoundRobin}(\vval[][(t)])$.\;
  Propose the allocation $A^{(t)}$.\;
  \lIf{$A^{(t)}$ is certified PROP1}{\Return $A^{(t)}$.}

  \BlankLine
  \tcp{\textbf{(b) Identify new separating constraints}}
  Receive a violating agent $i$.\;

  Create new constraints
  \(
    \mathcal{C}_i^{(t)} :=
    \begin{aligned}[t]
      &\Set{\val[i]}{v_i(A_i^{(t)}) \le \tfrac{1}{n} v_i(M)} \\
      &\cap \bigcap_{y \in M \setminus A_i^{(t)}}
          \Set{\val[i]}{v_i(A_i^{(t)} \cup \{y\}) \le \tfrac{1}{n} v_i(M)} \\
      &\cap \bigcap_{x \in A_i^{(t)}}
          \Set{\val[i]}{v_i(A_i^{(t)} \setminus \{x\}) \le \tfrac{1}{n} v_i(M)} .
    \end{aligned}
  \)    

  \BlankLine
  \tcp{\textbf{(c) Ellipsoid update}}
  $\mathcal{R}_i^{(t)} \gets \mathcal{R}_i^{(t-1)} \cap \mathcal{C}_i^{(t)}$.\;
  $\mathcal{R}_k^{(t)} \gets \mathcal{R}_k^{(t-1)}$ for all $k \neq i$.\;
}
\end{algorithm}

\begin{theorem}
\label{thm:prop1-ellipsoid-poly}
Under additive valuations in the mixed setting,
Algorithm~\ref{alg:ellipsoid-prop1} always terminates with a PROP1 allocation within $\poly(n,m)$ rounds of proposed allocations under the Single-Witness PROP1 Model, and runs in total time $\poly(n,m)$.
\end{theorem}

\begin{proof}[Proof sketch]
We again view each agent’s additive valuation as a point in a convex region of $\mathbb{R}^m$,
where $m$ is the number of items and each coordinate represents the value of one item.
The argument then parallels that for Theorem~\ref{thm:ef1-ellipsoid-poly}.

When the feedback consists of a violating agent~$i$ for the proposed allocation~$A^{(t)}$,
the algorithm constructs the set $\mathcal{C}_i^{(t)}$ consisting of the linear inequalities corresponding to the three types of PROP1 conditions for agent~$i$, 
as specified in Algorithm~\ref{alg:ellipsoid-prop1} step~(b).
By definition of a PROP1 violation, the true valuation $\vtrue[i]$ satisfies all inequalities in $\mathcal{C}_i^{(t)}$.

On the other hand, under the current \pval $\vval[i][(t)]$, the allocation $A^{(t)}$ produced by \textsc{DoubleRoundRobin} satisfies EF1, and hence also satisfies PROP1.
Therefore, at least one of the PROP1 violation inequalities in $\mathcal{C}_i^{(t)}$ is violated by $\vval[i][(t)]$.
It follows that intersecting the current feasible region with $\mathcal{C}_i^{(t)}$ eliminates the current \pval while preserving $\vtrue[i]$, yielding a valid separating hyperplane for agent~$i$.
All remaining (possibly redundant) inequalities in $\mathcal{C}_i^{(t)}$ only further restrict the feasible region,
without excluding the true valuation.

The ellipsoid update thus preserves feasibility while strictly shrinking the feasible region whenever the algorithm does not terminate.
Since under the true valuations, \textsc{DoubleRoundRobin} produces an EF1 allocation,
and EF1 implies PROP1, once a proposed allocation is consistent with the true valuations, the allocation satisfies PROP1, and the algorithm terminates.
Each iteration runs in $\poly(n,m)$ time, and each agent requires no more than $\poly(m)$
separating steps; summing over all agents yields a total running time and interaction complexity polynomial in~$n$ and~$m$.
\end{proof}

\subsection{General Ellipsoid-Based Fairness Framework}
\label{subsec:ellipsoid-framework}

We now abstract the common structure underlying the EF1 and PROP1 instantiations
to obtain a general \emph{ellipsoid-based framework}
for interaction-efficient fair allocation under limited feedback.
The framework consists of four modular components, each of which can be replaced or adapted to target different fairness notions and valuation classes.

\begin{description}
\item[(1) Allocation Subroutine.]
Given access to the true valuations from a specified valuation class,
this subroutine must output an allocation satisfying the target fairness notion in polynomial time.
For EF1 and PROP1 under additive valuations, we instantiate this module with the
\textsc{DoubleRoundRobin} algorithm of \citet{Aziz2022ef1mixed};
the \emph{Generalized Envy Graph} algorithm proposed in
\citet[Algorithm~2]{Aziz2022ef1mixed} can also serve as an alternative.
For PROP1, one may further use the polynomial-time procedure established by
\citet[Theorem~5]{Aziz2022ef1mixed}, which computes a contiguous PROP1 allocation. Here, ``contiguous'' means that the items are arranged in some arbitrary order,
and each agent receives a block of consecutive items in that order.
In the setting of identical monotone valuations with goods, an EF1 connected allocation can also be computed via Algorithm~1 of \citet{Bilo2022Connect}, and may be used as an allocation subroutine when contiguity is required.

\item[(2) Feedback Model.]
The feedback model specifies the form of feedback received in response to a fairness violation,
and thereby determines the information structure available to the algorithm.
Our framework applies whenever the feedback provides a \emph{witness}
that enables the algorithm to construct a valid separating constraint.
The single-witness EF1 and PROP1 models are examples of such feedback models.
More informative models, for instance those that allow direct EF1 interactions
between specific agents, also fit naturally into the framework.
In contrast, weaker models that merely report that some agent~$i$
violates EF1 without identifying a concrete witness are insufficient,
since they do not provide enough information to construct a separating hyperplane between the current and true valuations.

\item[(3) Separation construction rule.]
Given the feedback, the algorithm must specify how to translate the witness into one or more linear constraints that separate the current \pval from the true valuation. The precise rule depends on the fairness notion and adversary type.  
In our EF1 and PROP1 instantiations, each fairness violation  witness naturally induces a collection of linear inequalities,
corresponding to the relevant violated EF1 or PROP1 conditions.
These inequalities define a set of halfspace constraints that is consistent with the true valuations.
As shown in the correctness analysis, at least one constraint separates the current \pval from the true one, thus providing the separating hyperplane required by the ellipsoid method. 

\item[(4) Valuation Encoding.]
The valuation space must first admit a convex representation so that the ellipsoid update applies.  
The dimensionality of the ellipsoid then depends on how valuations are encoded. 
For additive valuations, the representation for each agent's valuations is $m$-dimensional, with one coordinate per item,
which leads to a polynomial-time implementation.
For richer valuation classes such as monotone valuations, 
the dimensionality typically increases substantially (for example, to $2^m$ if each subset is assigned an independent value).
In these cases, the ellipsoid method remains conceptually applicable, 
but the number of iterations is no longer polynomial in the natural parameters $n$ and $m$.
A somewhat natural restriction of this general model (here considered for goods) would be one in which agents are interested in bundles of at most $s$ items, for some constant $s$. Then, $v_i(A) = \max_{A' \subseteq A: |A'| \leq s} v_i(A')$ for all $A$ with $\SetCard{A} > s$. In this setting, the number of variables is $O(n m^s)$, and the algorithm can compute the valuation of $i$ for a bundle $A$ from the LP variables $v_i(A')$ (for $\SetCard{A'} \leq s$) in time $O(m^s)$. Another natural class of models can be obtained when the \emph{algorithm} only considers a polynomial-sized subset of possible bundles; we will see this in Section~\ref{sec:monotone-ef1}, where the algorithm chooses an (arbitrary, but fixed) ordering of items, and only allocates agents bundles which are the union of an interval plus at most one additional item; note that there are at most polynomially many such bundles whose valuations must be learned.
\end{description}

Together, these components form a unified ellipsoid-based framework that captures
a wide spectrum of fairness notions and feedback structures.
Whenever the allocation subroutine and adversary together admit a valid separating rule,
the framework achieves convergence in a number of interactions polynomial in the valuation dimension.
We further demonstrate in Appendix~\ref{sec:efx-charity} that this framework 
extends naturally to another relaxation of envy-freeness, namely \emph{EFX-with-charity}.

In the remainder of the paper, we return to \emph{EF1} as our primary fairness notion, and focus on refined algorithmic guarantees and structural insights specific to EF1.

\section{Polynomial-Interaction EF1 in the Single-Witness Model under Monotone Valuations}
\label{sec:monotone-ef1}
As discussed in the context of valuation encoding in the general framework,
the ellipsoid-based approach remains conceptually applicable to monotone valuations.
However, when valuations are represented by vectors,
a monotone valuation requires an encoding with exponentially many dimensions --- one for each subset of the items --- which in turn leads to exponentially many ellipsoid iterations and hence exponential interaction complexity.
This motivates a fundamentally different approach that avoids reasoning over such high-dimensional valuation representations.

Throughout this section, we restrict attention to the \emph{goods setting} with monotone valuations.
As noted in Section~\ref{sec:prelim}, monotonicity is defined in the opposite direction for chores,
and the mixed setting does not admit a well-defined notion of monotonicity.

Our strategy is to impose additional structure on the \emph{allocations} rather than on the valuations.
By restricting attention to allocations with simple geometric form,
we ensure that EF1 violations can arise only in a controlled number of ways,
which in turn allows the algorithm to make progress using only polynomially many single-witness EF1 interactions.

To this end, the algorithm puts the items in an arbitrary but fixed order, and considers only bundles that are ``simple'' with respect to this ordering.
The simplest such bundles are contiguous intervals, i.e.,
bundles consisting of items that occupy consecutive positions in the chosen ordering.
While interval EF1 allocations were previously known to exist in the special case of identical valuations~\citep{oh2021fairly}, \citet{Igarashi2023Discrete} proves that, for any fixed ordering, such allocations also exist under heterogeneous monotone valuations.
Our analysis below uses an independently obtained, slightly weaker structural result.
Building on the interval-EF2 existence result of
\citet[Theorem~5.4]{Bilo2022Connect}, 
we show that EF1 can be achieved with a slight relaxation of the interval structure.

For completeness, Section~5.1 presents our independently obtained result showing that every fair division instance with monotone valuations admits an EF1 allocation in which each agent receives a bundle obtained from an interval plus at most one additional item.
We then exploit this structural characterization in Section~5.2 to design an EF1 algorithm with polynomial interaction complexity under the Single-Witness EF1 Model.\footnote{Replacing our existence result with Igarashi's theorem~\citep[Theorem~3.1]{Igarashi2023Discrete} yields an analogous algorithm restricted to
contiguous allocations.}

\subsection{Existence of Almost-Contiguous EF1 Allocations}
\label{subsec:almost-contiguous}

We begin by formalizing the notion of an almost-contiguous allocation.
We will frequently reason about intervals of discrete numbers, and use the notation $\Interv{L}{R} = \Set{j \in \mathbb{Z}}{L \leq j \leq R}$, with the understanding that $\Interv{L}{R} = \emptyset$ when $R < L$.

\begin{definition}[Almost-Contiguous Allocation]
\label{def:almost-contiguous}
An allocation $A = (A_1,\dots,A_n)$ is said to be \emph{almost-contiguous} with respect to the given ordering of items if each agent $i$ receives a bundle $A_i$ that can be written as the union of a contiguous interval and at most one additional item.
Formally, for each agent $i$, there exist indices $j_{i,L}, j_{i,R}$ such that $\Interv{j_{i,L}}{j_{i,R}} \subseteq A_i$ and $|A_i \setminus \Interv{j_{i,L}}{j_{i,R}}| \leq 1$.
\end{definition}

\begin{theorem}[Existence of Almost-Contiguous EF1 Allocations for Goods]
\label{thm:exist-almost-contiguous}
{Let $(N, M,$ $(v_i)_{i\in N})$ be an arbitrary instance} of indivisible goods with monotone valuations. 
Then, for any fixed ordering of the items,
there exists an \emph{almost-contiguous allocation}  that is EF1.
\end{theorem}

\noindent
\emph{Proof overview.}
The proof builds on the existential argument of \citet[Theorem~5.4]{Bilo2022Connect}, which uses a carefully constructed representation of partial interval allocations together with Sperner's Lemma to establish the existence of a contiguous EF2 allocation under monotone valuations.
We modify their construction to obtain an almost-contiguous EF1 allocation.

\subsubsection{Background and proof intuition.}

The idea of allocating contiguous bundles is closely related to the classical problem of \emph{cake cutting}, which studies fair division of a divisible heterogeneous resource.
In this literature, the resource is typically modeled as the interval $[0,1]$, and each agent has their own valuation over measurable subsets.
A central goal is to partition the cake into pieces and allocate one piece to each agent such that the resulting allocation is envy-free.

A fundamental result in this area shows that envy-free cake divisions always exist.
Specifically, it is known that there always exists an envy-free allocation in which each agent receives a single interval of the cake \citep{stromquist1980cut}.
A classical approach to proving this result uses Sperner’s Lemma \citep{su1999rental}:
one considers the space of all interval partitions of the cake, labels them according to agents’ preferences, and applies a Sperner’s Lemma argument to establish the existence of an envy-free division.
Within this framework, the positions of the cuts between intervals play a central role.
These positions, often referred to as \emph{knife positions}, parameterize the space of possible allocations.

The indivisible-goods setting introduces additional challenges, as items cannot be split.
Building on the ideas underlying the cake-cutting proofs, \citet[Theorem~5.4]{Bilo2022Connect} establish an existence result for indivisible goods with monotone valuations.
At a high level, their argument considers a discrete analogue of interval partitions along a fixed ordering of the items. 
The knife positions are allowed to be either between two items or on an item (covering the item), and the construction allows the \emph{boundary items}, namely items covered by the knife positions, to remain temporarily unassigned.

By organizing these partial allocations into a lattice and applying Sperner’s lemma, \citet{Bilo2022Connect} show the existence of a collection of $n$ ``partial allocations'' that differ only in the assignment of the boundary items and agree on the assignment of the remaining interval bundles to agents.
These partial allocations have the property that there exists a one-to-one matching between agents and these partial allocations (i.e., the \emph{entire allocation of bundles to all agents}) such that, in the partial allocation assigned to agent $i$, agent $i$ is envy-free with respect to their own bundle.
A careful assignment of the remaining boundary items then yields a fully contiguous allocation that is EF2.

We build on this framework and modify the argument to obtain EF1, while slightly relaxing the requirement of contiguity.

\subsubsection{Knife positions and interval partitions.}
We omit the technical details of the application of Sperner’s lemma in the proof of \citet[Theorem~5.4]{Bilo2022Connect}, and instead work directly with the discrete structure that arises from their existence argument.

Fix an arbitrary ordering of the items, which we identify without loss of generality with positions $1,2,\ldots,m$ along a line.
A \emph{knife position} is a vector $\bm{k}=(k^1,\ldots,k^{n-1})$ specifying the locations of $n{-}1$ cuts between consecutive agents' bundles.
Each coordinate $k^j$ is chosen from the discrete grid
\[
\left\{\tfrac12, 1, \tfrac32, \ldots, m, m+\tfrac12\right\}.
\]
A half-integer value indicates a cut between two consecutive items, whereas an integer value indicates that the cut \emph{covers} the item at that position.

Given a knife position $\bm{k}$, it induces a partial interval partition
$P(\bm{k})=\left(P(\bm{k})_1,\ldots,P(\bm{k})_n\right)$ of the items as follows:
\[
\begin{aligned}
P(\bm{k})_1 &= \Interv{1}{\left\lfloor k^1-\tfrac12 \right \rfloor},\\
P(\bm{k})_j &= \Interv{\left\lceil k^{j-1}+\tfrac12 \right\rceil}{\left\lfloor k^j-\tfrac12 \right\rfloor}
&& (2 \le j \le n{-}1),\\
P(\bm{k})_n &= \Interv{\left\lceil k^{n-1}+\tfrac12 \right\rceil}{m}.
\end{aligned}
\]
If a coordinate $k^j$ is integral, the corresponding item is covered by the knife and does not belong to either adjacent bundle.
Such items are referred to as \emph{boundary items} and are treated as temporarily unassigned, as discussed above.

\medskip
To illustrate the knife construction, consider $n=4$ agents and $m=5$ items
labeled $1$--$5$ in order, and the knife position
$\bm{k}=(1,\,2.5,\,4)$.
The resulting partial interval partition is summarized in Table~\ref{tab:knife-example}.

\begin{table}[H]
\centering
\renewcommand{\arraystretch}{1.2}
\begin{tabular}{@{}l@{\quad}l@{\quad}l@{}}
\toprule
\textbf{Knife position} & \textbf{Visual representation of items} & \textbf{Bundles $P(\bm{k})$} \\
\midrule
$\bm{k}=(1,\,2.5,\,4)$ &
$\boxed{1}\;2\;|\;3\;\boxed{4}\;5$ &
$\varnothing,\;\{2\},\;\{3\},\;\{5\}$ \\
\bottomrule
\end{tabular}
\caption{
An illustration of a knife position and the induced partial interval partition.
A vertical bar “\textbar{}” denotes a half-integer knife placed between two consecutive items,
while a boxed item indicates an integer knife covering that item.
Covered items belong to neither adjacent bundle and are treated as boundary items.
}
\label{tab:knife-example}
\end{table}

\subsubsection{Families of knife positions.}
We now describe the specific family of knife positions that arises in the existence proof of \citet{Bilo2022Connect}.

Fix any valid knife position (i.e., vector) $\bm{k}_1$.
From this initial position, we define a family $\bm{k}_1,\bm{k}_2,\ldots,\bm{k}_n$ of knife positions as follows.
For each $i \in \SET{1,\ldots,n-1}$, the position $\bm{k}_{i+1}$ is obtained from $\bm{k}_i$ by shifting the \Kth{i} knife by a half unit to the right, while leaving all other knives unchanged:
\[
k_{i+1}^j =
\begin{cases}
  k_i^j + \tfrac12, & \text{if } j=i,\\[2pt]
  k_i^j, & \text{otherwise.}
\end{cases}
\]

We illustrate this construction by returning to the example introduced above with $n=4$ agents and $m=5$ items, and taking $\bm{k}_1=(1,\,2.5,\,4)$.
The resulting family of knife positions and the induced partial partitions are shown in Table~\ref{tab:knife-family-example}.

\begin{table}[H]
\centering
\renewcommand{\arraystretch}{1.2}
\begin{tabular}{@{}l@{\quad}l@{\quad}l@{}}
\toprule
\textbf{Knife position} & \textbf{Visual representation of items} & \textbf{Bundles $P(\bm{k}_i)$} \\
\midrule
$\bm{k}_1 = (1,\,2.5,\,4)$ &
$\boxed{1}\;2\;|\;3\;\boxed{4}\;5$ &
$\varnothing,\;\{2\},\;\{3\},\;\{5\}$ \\[4pt]

$\bm{k}_2 = (1.5,\,2.5,\,4)$ &
$1\;|\;2\;|\;3\;\boxed{4}\;5$ &
$\{1\},\;\{2\},\;\{3\},\;\{5\}$ \\[4pt]

$\bm{k}_3 = (1.5,\,3,\,4)$ &
$1\;|\;2\;\boxed{3}\;\boxed{4}\;5$ &
$\{1\},\;\{2\},\;\varnothing,\;\{5\}$ \\[4pt]

$\bm{k}_4 = (1.5,\,3,\,4.5)$ &
$1\;|\;2\;\boxed{3}\;4\;|\;5$ &
$\{1\},\;\{2\},\;\{4\},\;\{5\}$ \\
\bottomrule
\end{tabular}
\caption{
The family of knife positions generated from $\bm{k}_1=(1,\,2.5,\,4)$ in the example introduced above.
The visual representation uses the same notation as in Table~\ref{tab:knife-example}.
}
\label{tab:knife-family-example}
\end{table}

We now state the key guarantee for the family of knife positions that is implicit in the existence proof of \citet{Bilo2022Connect} and will be crucial for our construction.

\begin{lemma}[Derived from the proof of {\citet[Theorem~5.4]{Bilo2022Connect}}]
\label{lem:knife-preference}
For every instance of pure goods with monotone valuations and any fixed ordering of the items,
there exists a starting knife position $\vc{k}_1$ such that, for the induced partial partitions $P(\vc{k}_1),\dots,P(\vc{k}_n)$, one can assign to each agent $i$ a distinct partial partition $\varphi(i)$ and a distinct bundle $\sigma(i)$ within the partial partition $\varphi(i)$ so that $i$ prefers the assigned bundle to all other bundles in the assigned partial allocation. 

Formally, there exist bijections $\varphi,\sigma: [n] \to [n]$ such that, for every $i \in [n]$,
\[
  v_i\!\left(P\!\left(\vc{k}_{\varphi(i)}\right)_{\sigma(i)}\right)
  \; \ge \;
  v_i\!\left(P\!\left(\vc{k}_{\varphi(i)}\right)_{j}\right)
  \quad\text{for all } j \in [n].
\]
\end{lemma}
We refer to $P\!\left(\vc{k}_{\varphi(i)}\right)_{\sigma(i)}$ as \emph{agent~$i$’ s designated bundle}.
The family of designated bundles 
\[ \Set{P\!\left(\bm{k}_{\varphi(i)}\right)_{\sigma(i)}}{i \in [n]}
\] forms a (partial) partition of~$M$; that is, the bundles are pairwise disjoint, but do not necessarily cover all of~$M$.

\subsubsection{Basic partition and boundary items.}
Based on these $n$ knife positions, \citet{Bilo2022Connect} define the \emph{basic bundle} of agent $i$ as the set of items which $i$ receives under \emph{all} of the partial allocations, i.e.,
\[
\bp[\sigma(i)] \;= \; \bigcap_{j=1}^n P(\vc{k}_j)_{\sigma(i)}.
\]
The resulting partial partition $\bp = (\bp[\sigma(1)], \ldots, \bp[\sigma(n)])$ is then referred to as the \emph{basic partition}.
The remaining items are referred to as \emph{boundary items} because they must lie on the boundary between bundles:
\[
R \;:=\; M \setminus \bigcup_{i=1}^n \bp[\sigma(i)].
\]
Equivalently, $R$ can be characterized directly in terms of the initial knife position~$\bm{k}_1$ as
\[
R
\;=\;
\Set{
  k_1^j}{k_1^j \in \mathbb{Z},\; k_1^j \le m}
\;\cup\;
\Set{
  k_1^j+\tfrac{1}{2}}{k_1^j \in \tfrac{1}{2}+\mathbb{Z},\; k_1^j+\tfrac{1}{2} \le m}.
\]

As an immediate consequence of this alternate characterization, each knife contributes at most one boundary item:
it contributes the item it covers if its coordinate is integral,
and the item immediately to its right if its coordinate is a half-integer,
since this item will eventually be covered in $\vc{k}_n$ by the movement rule of the knives.
Hence,
\[
|R| \;\le\; n-1.
\]
Moreover, as observed by~\citet{Bilo2022Connect},
each boundary item can appear in the bundles of at most one side (either in the bundle immediately to its left or in the bundle immediately to its right) throughout all knife positions $\vc{k}_1,\dots,\vc{k}_n$.
In particular, no boundary item ever switches sides across different knife configurations.

\medskip
\noindent
\emph{Example (continued).} Continuing the example with $n=4$ agents, $m=5$ items, and the knife-position family generated from $\bm{k}_1 = (1,\,2.5,\,4)$ as shown above, the resulting basic partition and boundary items are summarized in Table~\ref{tab:knife-basic-example}.

\begin{table}[H]
\centering
\renewcommand{\arraystretch}{1.2}
\begin{tabular}{@{}l@{\quad}l@{\quad}l@{}}
\toprule
\textbf{Knife position} & \textbf{Visual representation of items} & \textbf{Bundles $P(\bm{k}_i)$} \\
\midrule
\textbf{Basic partition $\bp$} &
--- &
$\varnothing,\;\{2\},\;\varnothing,\;\{5\}$ \\[4pt]

\textbf{Boundary items $R$ (boxed)} &
$\boxed{1}\;2\;\boxed{3}\;\boxed{4}\;5$ &
--- \\
\bottomrule
\end{tabular}
\caption{
Basic partition and boundary items corresponding to the knife-position family
illustrated in Table~\ref{tab:knife-family-example}.
}
\label{tab:knife-basic-example}
\end{table}

We now relate the designated bundles identified in Lemma~\ref{lem:knife-preference}
to the basic partition and the boundary items defined above.

\begin{lemma}\label{lem:designated-bundle}
For every agent $i \in [n]$, let $P(\vc{k}_{\varphi(i)})_{\sigma(i)}$ denote $i$'s designated bundle,
and let $\bp$ and $R$ be defined as above.
Then each designated bundle satisfies one of the following two conditions:
\begin{enumerate}
  \item $P(\vc{k}_{\varphi(i)})_{\sigma(i)} = \bp[\sigma(i)]$, or
  \item $P(\vc{k}_{\varphi(i)})_{\sigma(i)} = \bp[\sigma(i)] \cup \{r_i\}$ for some boundary item $r_i \in R$.
\end{enumerate}
\end{lemma}

\begin{proof}
\newcommand{\iteml}{r_{\mathrm{L}}}
\newcommand{\itemr}{r_{\mathrm{R}}}
\newcommand{\knifel}{\kappa_{\mathrm{L}}}
\newcommand{\knifer}{\kappa_{\mathrm{R}}}

By construction, each designated bundle $P(\vc{k}_{\varphi(i)})_{\sigma(i)}$
can contain at most the items of its corresponding basic bundle
$\bp[\sigma(i)]$ together with
the boundary items that lie immediately adjacent to it in the item ordering.
That is, it may include at most one boundary item to the left and one to the right of $\bp[\sigma(i)]$.

If $\sigma(i)=1$ or $\sigma(i)=n$, then the designated bundle has only one adjacent side and therefore can contain at most one boundary item.
It thus suffices to show that no \emph{interior} designated bundle, that is, one with $\sigma(i)\in\{2,\ldots,n-1\}$, can include boundary items from both sides.
Assume for contradiction that there exists an agent $i$ with
\[
P(\vc{k}_{\varphi(i)})_{\sigma(i)}
\;=\;
\SET{\iteml, \itemr} \cup \bp[\sigma(i)],
\]
where $\iteml$ and $\itemr$ are, respectively, the boundary items immediately to the left and to the right of $\bp[\sigma(i)]$.
Let $\knifel :=\sigma(i)-1$ be the (index of the) knife separating bundles $\sigma(i)-1$ and $\sigma(i)$, and let $\knifer := \sigma(i)$ be the (index of the) knife separating bundles $\sigma(i)$ and $\sigma(i)+1$.

By the construction, knife coordinates lie on the half-integer grid and the knives move, in order of their indices, by a half unit per step.
Because we assumed that $\SET{\iteml, \itemr} \subseteq P(\vc{k}_{\varphi(i)})_{\sigma(i)}$, at the initial configuration $\vc{k}_1$, we must have
\[
k_1^{\knifel} \;=\; \iteml - \tfrac{1}{2}
\qquad\text{and}\qquad
k_1^{\knifer} \;=\; \itemr;
\]
that is, the knife $\knifel$ sits between $\iteml-1$ and $\iteml$, while $\knifer$ covers $\itemr$.

For $\iteml$ to belong to $P(\vc{k}_{\varphi(i)})_{\sigma(i)}$ at configuration $\vc{k}_{\varphi(i)}$, the left knife $\knifel$ must still be strictly to the left of $\iteml$ (i.e., at $\iteml-\tfrac{1}{2}$); that is, it must not have moved from its initial position.
For $\itemr$ to belong to the same bundle at $\vc{k}_{\varphi(i)}$, the right knife $\knifer$ must already have advanced by a half step from $\itemr$ to $\itemr+\tfrac{1}{2}$, thereby releasing $\itemr$ into the bundle.
This contradicts the order in which the knives move, showing that a designated bundle cannot simultaneously contain boundary items from both sides.

Therefore, every designated bundle $P(\vc{k}_{\varphi(i)})_{\sigma(i)}$ contains at most one boundary item.
\end{proof}

\subsubsection{Proof of Theorem~\ref{thm:exist-almost-contiguous}}

We are now ready to prove Theorem~\ref{thm:exist-almost-contiguous}.

\begin{proof}[Proof of Theorem~\ref{thm:exist-almost-contiguous}]
By Lemma~\ref{lem:designated-bundle}, each agent’s designated bundle consists of their basic bundle together with at most one boundary item.
Some designated bundles already contain a boundary item, while others consist solely of a basic bundle.

Recall that there are at most $n-1$ boundary items in total.
Each boundary item that already appears in a designated bundle is associated with a unique agent.
Consequently, after accounting for these agents,
the number of remaining boundary items is strictly smaller than the number of agents whose designated bundles contain no boundary item.
We can therefore assign each remaining boundary item to a distinct agent whose designated bundle is a basic bundle.
This yields an almost-contiguous allocation (Definition~\ref{def:almost-contiguous}) $A = (A_1,\ldots,A_n)$ of the item set~$M$, in which each agent receives their designated bundle and possibly one additional item from $R$.

We now verify that this allocations is EF1.
By Lemma~\ref{lem:knife-preference}, for each agent~$i$,
\[
  v_i\!\left(P\!\left(\vc{k}_{\varphi(i)}\right)_{\sigma(i)}\right)
  \;\ge\;
  v_i\!\left(P\!\left(\vc{k}_{\varphi(i)}\right)_{\sigma(j)}\right)
  \quad\text{for all } j \in [n].
\]
Since every basic bundle $\bp[\sigma(j)]$ is contained in $P(\vc{k}_t)_{\sigma(i)}$ for all knife positions~$\vc{k}_t$, it follows that
\[
  v_i\!\left(P\!\left(\vc{k}_{\varphi(i)}\right)_{\sigma(i)}\right)
  \;\ge\;
  v_i\!\left(\bp[\sigma(j)]\right)
  \quad\text{for all } i,j \in [n].
\]
Thus, each agent weakly prefers their designated bundle
to every basic bundle.

Consider any two agents $i$ and $j$.
If agent~$j$ receives only a basic bundle, then agent~$i$ does not envy~$j$.
Otherwise, agent~$j$ receives a bundle of the form $\bp[\sigma(j)] \cup \{r\}$ for some boundary item~$r$.
Removing this single item yields the basic bundle $\bp[\sigma(j)]$, which agent~$i$ does not envy by the above inequality.
Hence, agent~$i$ is EF1 with respect to~$j$.

Combining the construction of the allocation with the EF1 analysis above, we conclude that the allocation $A$ is an almost-contiguous EF1 allocation.
This completes the proof.
\end{proof}

\subsection{A Polynomial-Interaction Algorithm via Almost-Contiguous Allocations}
\label{subsec:polyquery-almost-contiguous}

We now return to the ellipsoid-based elimination framework developed in Section~\ref{subsec:ellipsoid-framework}, and show how it can be instantiated to obtain a polynomial-interaction algorithm for EF1 under general monotone valuations.
The key new ingredient is to restrict the allocation subroutine to the class of \emph{almost-contiguous EF1 allocations} introduced in Section~\ref{subsec:almost-contiguous},
whose existence is guaranteed by Theorem~\ref{thm:exist-almost-contiguous}.
At a high level, the algorithm proceeds exactly as Algorithm~\ref{alg:ellipsoid-ef1}.
The only differences are the following:
\begin{itemize}
\item In step~(a), the ellipsoid method maintains a \pval profile $\vval$, as before.
Given $\vval$, the allocation subroutine of \textsc{DoubleRoundRobin} is replaced with the allocation routine described in ~\Cref{subsec:almost-contiguous}, which returns an EF1 allocation that is almost-contiguous.

\item Since we now focus on the pure goods setting with monotone valuations, the initialization of the ellipsoid is modified accordingly: the domain of all valuations is now $[0,1]$.

Moreover, while monotone valuations are fully characterized only by the agents' valuations for every subset of items, we exploit the fact that the allocation subroutine only ever produces almost-contiguous allocations.
Consequently, it suffices to maintain value variables only for almost-contiguous bundles.
Since there are (at most) $m \cdot (m+1)m/2$ such almost-contiguous bundles\footnote{By instead invoking Igarashi's theorem~\citep[Theorem~3.1]{Igarashi2023Discrete}, the allocation subroutine may be restricted to contiguous EF1 allocations; 
in that case, it suffices to maintain only \(O(m^2)\) interval values per agent.} (i.e., $m$ choices for the extra item and $(m+1)m/2$ contiguous intervals defined by their two endpoints), the initial ellipsoid is taken to be the hypercube $[0,1]^{m \cdot (m+1)m/2}$.
\end{itemize}

\begin{theorem}[Ellipsoid-Based Polynomial-Interaction EF1 for Monotone Goods]
\label{thm:ellipsoid-polyquery-monotone}
Under monotone valuations in the pure goods setting,
Algorithm~\ref{alg:ellipsoid-ef1},
instantiated with the almost-contiguous EF1 allocation subroutine described above,
terminates after at most \(O(n m^6)\) single-witness EF1 interactions and outputs an EF1 allocation.\footnote{Using Igarashi's theorem~\citep[Theorem~3.1]{Igarashi2023Discrete}, the analogous counting argument yields an \(O(nm^4)\) interaction bound when the allocation subroutine is restricted to contiguous allocations.}
\end{theorem}

\begin{proof}
Fix an arbitrary ordering of the items as $1,2,\dots,m$ along the line.
By Theorem~\ref{thm:exist-almost-contiguous}, there exists an almost-contiguous EF1 allocation under any valuation profile.

We use a counting argument, showing that each round in which the single-witness EF1 feedback returns a violation rules out a distinct local ``sub-configuration'' of an almost-contiguous allocation.
Let $A$ be the almost-contiguous allocation proposed in some round, and suppose that the feedback revealed an ordered violating pair $(i,j)$ such that \(A_i \prec^{\mathrm{EF1}}_i A_j\).
Write
\[
A_i = B_i \cup R_i
\qquad\text{and}\qquad
A_j = B_j \cup R_j,
\]
where $B_i,B_j$ are the (possibly empty) basic contiguous bundles and $R_i, R_j$ are the (possibly empty) boundary items (singletons).
The witness feedback certifies that, under the true valuations, agent $i$ does not EF1-prefer their own bundle to $j$'s bundle.
Equivalently, the feedback implies the following structural constraint: whenever agent~$i$ receives exactly the bundle $B_i \cup R_i$, no other agent can receive the bundle $B_j \cup R_j$ in any almost-contiguous allocation that is EF1 under valuations consistent with all feedback observed so far.
This information is stored as an ordered local \emph{forbidden tuple}
\[
(i,\, B_i,\, B_j,\, R_i,\, R_j)
\]
that cannot occur in any almost-contiguous allocation that is EF1 under the true valuations.

We therefore associate each violation with a forbidden tuple of the above form and collect all such tuples in a set~$\mathcal{B}$, which we refer to as the \emph{blacklist}.
Because Algorithm~\ref{alg:ellipsoid-ef1} updates the feasible region using the separating constraint implied by the witness feedback, the ellipsoid process makes non-redundant progress: the same forbidden tuple cannot be witnessed again in a later round.
Consequently, the total number of proposed allocations is at most the number of distinct tuples of the form $(i,B_i,B_j,R_i,R_j)$.

It remains to bound this number.
There are at most \(O(m^2)\) choices for a contiguous basic bundle (identified by its two endpoints, allowing the empty bundle), and at most \(O(m)\) choices for a boundary item (including the empty choice).
Thus the number of distinct tuples is at most
\[
O\!\left(n \cdot m^2 \cdot m^2 \cdot m \cdot m\right)
\;=\; O(nm^6).
\]
Therefore, after at most \(O(nm^6)\) single-witness EF1 interactions, the algorithm must reach a round in which the proposed allocation is certified as EF1.\footnote{A somewhat weaker polynomial upper bound on the number of interactions can be obtained by simply noting that the total number of variables encoding the valuations is $O(nm^3)$. The properties of the ellipsoid method alluded to before imply termination in a polynomial number of iterations. We introduce the blacklist-based analysis in part because it serves as the point of departure for a different approach discussed in Appendix~\ref{sec:direct-almost-contiguous}.}
\end{proof}

Note that Theorem~\ref{thm:ellipsoid-polyquery-monotone} only bounds the number of proposed allocations which the algorithm makes before it finds an EF1 allocation with respect to the true valuations. It does not bound the running time of the algorithm's internal computation.

From a computational perspective, restricting attention to almost-contiguous allocations circumvents the primary obstacle in the monotone setting, namely the exponential-size valuation encoding.
In particular, when attention is restricted to almost-contiguous bundles, the dimension of the valuation representation becomes polynomial in~$m$.
Nevertheless, this restriction introduces a different computational challenge.
When we additionally require the allocation to be almost-contiguous,
computing an EF1 allocation becomes a nontrivial task.
Recall that the existence proof for Theorem~\ref{thm:exist-almost-contiguous} was non-constructive, and relied on Sperner's Lemma. 
The proof gives rise to a search procedure, by considering all starting knife positions $\vc{k}_1$ and their corresponding sequences of knife positions and allocations. 
However, this search would take exponential time. 
Whether almost-contiguous EF1 allocations can be found in polynomial time remains open,\footnote{The proof of \citet[Theorem~3.1]{Igarashi2023Discrete} is likewise non-constructive, and the corresponding computational question for contiguous allocations remains open.} paralleling an open question of \citet{Bilo2022Connect}: whether interval EF2 allocations can be found in polynomial time. 
Indeed, given the nature of the existence proof, it is quite conceivable that the search problem could be PPAD-complete.

Motivated by this observation, one could try to eliminate the ellipsoid method altogether and instead directly keep track of ``blacklisted'' almost-contiguous EF1 allocations, only proposing allocations which have not been blacklisted.
However, this shift comes at a different cost.
As we show in Appendix~\ref{sec:direct-almost-contiguous} using a reduction from \textsc{Independent Set}, by directly focusing on interval-based allocations, we replace a computational subproblem whose complexity is open with a problem that is NP-complete in general.\footnote{The reduction already establishes NP-completeness even for contiguous allocations, so the same computational obstacle remains when the algorithm is instantiated using Igarashi's theorem~\citep[Theorem~3.1]{Igarashi2023Discrete}.}

\section{The Membership EF1 Model}
\label{sec:weak}
In contrast to the Single-Witness Model, where the adversary provides witnesses of fairness violations,
the Membership Model reveals only a single bit of feedback indicating whether the proposed allocation satisfies the target fairness notion.
This minimal feedback severely limits the learner’s information gain per proposed allocation.

In this section, we establish that for EF1 even under the extremely restricted class of identical additive binary valuations,
any algorithm operating in the Membership Model must make exponentially many allocation proposals in the number of agents $n$.
We also present positive results demonstrating that when the number of agents \(n\) is fixed,
EF1 can still be achieved with a polynomial number of membership interactions.

\subsection{Exponential Interaction Lower Bound}
\label{subsec:weak-lower-bound}

Recall that in the additive setting, each valuation satisfies 
\(v_i(S) = \sum_{x \in S} v_i(x)\) for all \(S \subseteq M\).
In this subsection, we focus on the highly restricted class of
\emph{identical binary additive} valuations,
where \(v_i(x) = v(x) \in \SET{0,1}\) for every agent \(i\) and item \(x\).

We show that the Membership Model still requires exponentially many allocation proposals in $\min(m-n,n)$ in order to find an EF1 allocation.
In particular, when $m \ge 2n$, the required number of allocation proposals is exponential in $n$.
Our lower bound holds for any $m>n$.
(The case $m\le n$ is trivial,
as assigning each agent (at most) one distinct good yields an EF1 allocation.)

Our proof uses the probabilistic method. 
In the Membership Model, every proposed allocation receives
the same yes-or-no type of response, so a deterministic algorithm inevitably commits to a fixed sequence of proposed allocations until an EF1 allocation is found. 
If the algorithm tries fewer than exponentially many allocation proposals, we show that there exists a binary additive valuation for which none of these proposed allocations is EF1.
Hence, no deterministic procedure can succeed with fewer proposed allocations. 
By standard arguments, the same conclusion extends to randomized algorithms as well, implying that exponential interaction complexity is unavoidable in the Membership Model.

\begin{theorem}[Exponential interaction lower bound under i.i.d.~Bernoulli valuations]
\label{thm:weak-exp-iid}
Consider $n$ agents and $m>n$ goods under identical binary additive valuations.
Let $\Gamma$ be the distribution over valuations $v$ under which each item independently has value $1$ with probability $1/2$ (and $0$ otherwise).
Let $T={\frac{1}{\lfloor m/n\rfloor+1}} \left(\tfrac{4}{3}\right)^{\min(m-n,n)}$.
Then for any (possibly randomized) algorithm $\ALG$ that makes fewer than $T$ allocation proposals,
\[
\Prob[\val \sim \Gamma, \ALG]{\text{$\ALG$ finds an EF1 allocation under $\val$}} \;<\; 1.
\]
Consequently, to guarantee EF1 on \emph{all} valuations, any algorithm must make at least $T$ allocation proposals in the worst case.
\end{theorem}

Much of the tedious technical work is encapsulated by the following lemma, which captures that without loss of generality, an allocation maximizing the probability of being EF1 contains no empty bundles, and as few singleton bundles as possible.
The proof of this lemma is given in Section~\ref{sec:proof-balanced-allocation}.

\begin{lemma}
\label{lem:balanced-allocation}
Let $A=(A_1, \ldots, A_n)$ be any allocation. There exists an allocation $A'=(A'_1, \ldots, A'_n)$ such that all $A'_i \neq \emptyset$, and $|\Set{i}{|A_i| = 1} | = \max(0,2n-m)$, and such that
\[ \
  \Prob[\val \sim \Gamma]{A' \text{ is EF1 }} \geq \Prob[\val \sim \Gamma]{A \text{ is EF1 }}.
\]
\end{lemma}

The following technical lemma will be used to upper-bound the probability that a single allocation is EF1 under identical binary valuations.
The proof of this lemma is given in Appendix~\ref{sec:binom-two-point}.

\begin{lemma}[Two-point anti-concentration bound for $\mathrm{Binomial}(s,\tfrac12)$]
\label{lem:binom-two-point}
Let $X\sim \mathrm{Binomial}(s,\tfrac12)$ with $s\ge 2$.
Then for every integer $k\in\SET{0,1,\ldots,s}$,
\[
\Prob{X\in\SET{k,k+1}} \;\le\; \tfrac34 .
\]
\end{lemma}

\begin{proof}[Proof of Theorem~\ref{thm:weak-exp-iid}]
Fix any proposed allocation $A=(A_1,\dots,A_n)$. 
We will upper-bound $\Prob[\val \sim \Gamma]{A \text{ is EF1 }}$.
By Lemma~\ref{lem:balanced-allocation}, without loss of generality, we can focus on the case that $A_i \neq \emptyset$ for all $i$ and there exist only $\max(0,2n-m)$ indices $i$ with $|A_i| = 1$, i.e., $\max(0,2n-m)$ agents receive singleton bundles. 
For each $i$, let $s_i=|A_i|$, and let $X_i=v(A_i)$ be the number of $1$-valued items in $A_i$.
Under $\Gamma$, $X_i \sim \mathrm{Binomial}(s_i,\tfrac12)$, and because the $A_i$ are disjoint, the $X_i$ are independent across $i$.

Notice that EF1 under identical binary valuations is equivalent to $\max_i X_i-\min_i X_i\le 1$, or, equivalently, that there exists a $k$ such that $X_i \in \SET{k, k+1}$ for all $i$.
We will bound the probability of this event separately for all values of $k$.

First, we observe that when $k > \min_i s_i$, we have $\Prob{X_i \in \SET{k, k+1} \text{ for all } i} = 0$, because the agent receiving $\min_i s_i$ items cannot obtain $k$ or more valuable items.
Because $\min_i s_i \leq \lfloor m/n \rfloor$, it suffices to focus on values $k \leq \lfloor m/n \rfloor$.

For the remaining cases $k \le \lfloor m/n \rfloor$, we apply Lemma~\ref{lem:binom-two-point} to each bundle $A_i$ with $|A_i| \ge 2$, of which there must be at least $\min(m-n, n)$ by our assumption on the allocation $A$, and obtain that for each fixed integer $k$:

\[
\Prob[\val \sim \Gamma]{X_{i} \in \SET{k,k+1} \text{ for all } i}
\; \le \;
\left(\tfrac{3}{4}\right)^{\min(m-n,n)}.
\]

Because EF1 requires this event to occur for some value of $k$, and the events for different $k$ are disjoint, we obtain that
\begin{align*}
\Prob[\val \sim \Gamma]{A\text{ is EF1}}
& =
\sum_{k=0}^{\lfloor m/n \rfloor}
\Prob[\val \sim \Gamma]{X_i \in \SET{k,k+1} \text{ for all } i}
\; \le \;
(\lfloor m/n\rfloor+1) \cdot
\left(\tfrac{3}{4}\right)^{\min(m-n,n)}.
\end{align*}

If the algorithm makes $t$ allocation proposals, by the union bound (which does not require independence across proposals),
\[
\Prob[\val \sim \Gamma]{\text{some proposed allocation is EF1}}
\;\le\;
t\cdot (\lfloor m/n\rfloor+1) \cdot \left(\tfrac{3}{4}\right)^{\min(m-n,n)}.
\]
Thus any algorithm with $t < T = \frac{1}{\lfloor m/n\rfloor+1} \left( \tfrac{4}{3}\right)^{\min(m-n,n)}$ succeeds with probability strictly less than~$1$ under this distribution.
In particular, this implies that for any fixed deterministic algorithm that makes fewer than $T$ allocation proposals, there exists a valuation $\val$ in the support of the distribution for which none of the algorithm's proposed allocations are EF1.

For randomized algorithms $\ALG$, the same upper bound applies to the success probability by Yao’s minimax principle: since every deterministic strategy with $t<T$ succeeds with probability less than 1 under $\Gamma$, so does any randomized strategy. 
Hence there exists a (fixed) valuation $\val$ such that, with positive probability, $\ALG$ does not output an EF1 allocation in any round. Consequently, to guarantee EF1 on \emph{all} valuations, one needs at least $T$ allocation proposals.
\end{proof}

\emcomment{Possible this can go to the appendix?}
\subsubsection{Proof of Lemma~\ref{lem:balanced-allocation}}
\label{sec:proof-balanced-allocation}

\begin{proof}[Proof of Lemma~\ref{lem:balanced-allocation}]

As in the proof of Theorem~\ref{thm:weak-exp-iid}, let $s_i=|A_i|$.
If the allocation $A$ already contains no empty bundle, and as few singleton bundles as possible, then the lemma holds trivially, with $A'=A$. So we focus on the cases when $A$ contains an empty bundle, or more singleton bundles than necessary.

First assume that $A$ contains an empty bundle.
Since $m > n$, by the pigeonhole principle, there exists an agent $j$ such that $|A_j| \ge 2$.
Construct a new allocation $A'=(A'_1,\dots, A'_n)$ by moving an arbitrary item from $A_j$ to $A_i$, and leaving all other bundles unchanged.
Let $X_\ell=v(A_\ell)$ and $X'_\ell=v(A'_\ell)$ denote the bundle values under $A$ and $A'$, respectively.

Since $A_i=\emptyset$, we have $X_i=0$.
Therefore, under identical binary valuations, $A$ is EF1 if and only if $X_\ell \in \SET{0,1}$ for every $\ell$, so
\[
\Prob[\val \sim \Gamma]{A \text{ is EF1}}
=
\Prob[\val \sim \Gamma]{ X_\ell \in \SET{0,1} \text{ for all } \ell}.
\]
By independence of the random variables $X_\ell$ across agents, we obtain
\begin{align*}
\Prob{A \text{ is EF1}}
&=
\Prob{X_i\in\SET{0,1}}
\cdot
\Prob{X_j\in\SET{0,1}}
\cdot
\prod_{\ell\neq i,j}\Prob{X_\ell\in\SET{0,1}}\\
&=
1
\cdot
\Prob{\mathrm{Binomial}(s_j,\tfrac12)\in\SET{0,1}}
\cdot
\prod_{\ell\neq i,j}\Prob{X_\ell\in\SET{0,1}}\\
& \stackrel{(*)}{\le}
1
\cdot
\Prob{\mathrm{Binomial}(s_j-1,\tfrac12)\in\SET{0,1}}
\cdot
\prod_{\ell\neq i,j} \Prob{X_\ell\in\SET{0,1}}\\
&=
\Prob{X'_i \in \SET{0,1}}
\cdot \Prob{X'_j \in \SET{0,1}}
\cdot
\prod_{\ell \neq i,j} \Prob{X'_\ell\in\SET{0,1}}\\
&=
\Prob{X'_\ell \in \SET{0,1} \text{ for all } \ell} \\
&\le\;
\Prob{A' \text{ is EF1}}.
\end{align*}
In the step labeled (*), we used that due to $s_j \geq 1$,
\[
\Prob{\mathrm{Binomial}(s_j,\tfrac12)\in\SET{0,1}}
= 2^{-s_j} \cdot (s_j + 1)
\leq 2^{-(s_j-1)} \cdot s_j
= \Prob{\mathrm{Binomial}(s_j-1,\tfrac12)\in\SET{0,1}}.
\]
Applying this transformation repeatedly, while there exists an empty bundle, we observe that any allocation containing an empty bundle is dominated, in terms of EF1 probability, by another allocation with no empty bundles.

Next we show that among allocations with no empty bundles, having too many singleton bundles can only decrease the probability of the allocation being EF1. Specifically, we show that if an allocation has more than $\max(0,2n-m)$ singleton bundles, then it is dominated (in EF1 probability) by another allocation with no empty bundles and fewer singleton bundles.

Assume that $A$ has no empty bundles but strictly more than $\max(0,2n-m)$ singleton bundles.
Then there exists an agent $i$ with $|A_i|=1$, and by the pigeonhole principle, there exists another agent $j$ with $|A_j|=s_j\ge 3$.
Construct $A'$ by moving an arbitrary item from $A_j$ to $A_i$, leaving all other bundles unchanged.
Again, let $X_\ell=v(A_\ell)$ and $X'_\ell=v(A'_\ell)$.

Because $A$ contains at least one singleton bundle, it can only be EF1 if all bundle values lie in $\SET{0,1}$ or all bundle values lie in $\SET{1,2}$.
On the other hand, all bundle values lying in $\SET{0,1}$ or all bundle values lying in $\SET{1,2}$ is sufficient for $A'$ being EF1. 
The analysis is slightly complicated by the fact that the events of all bundle values lying in $\SET{0,1}$ is not disjoint from the event of all bundle values lying in $\SET{1,2}$ (when all bundle values are equal to 1).
To make the analysis precise, we define the following three events:
\begin{align*}
\Event[1]{E} & = \Set{v}{X_{\ell} = 1 \text{ for all } \ell \neq i, j}, \\
\Event[1,2]{E} & = \Set{v}{X_{\ell} \in \SET{1,2} \text{ for all } \ell \neq i, j} \setminus \Event[1]{E}, \\
\Event[0,1]{E} & = \Set{v}{X_{\ell} \in \SET{0,1} \text{ for all } \ell \neq i, j} \setminus \Event[1]{E}.
\end{align*}

Because $|A_i| = 1$, and using the independence of the valuations for distinct items, we can bound the probabilities of either allocation being EF1 as follows:

\begin{align*}
\Prob{A \text{ is EF1}}
= & \phantom{+}
\Prob{\Event[1]{E}} \cdot 
\left( 
\Prob{X_i = 0} \cdot \Prob{X_j \in \SET{0,1}}
+ \Prob{X_i = 1} \cdot \Prob{X_j \in \SET{0,1,2}} \right)
\\ & + 
\Prob{\Event[0,1]{E}} \cdot
\Prob{X_i \in \SET{0,1}} \cdot
\Prob{X_j \in \SET{0,1}}
\\ & + 
\Prob{\Event[1,2]{E}} \cdot
\Prob{X_i = 1} \cdot
\Prob{X_j \in \SET{1,2}}
\\ = & \phantom{+}
\Prob{\Event[1]{E}} \cdot 
\left( \Prob{X_j \in \SET{0,1}} + \tfrac12 \cdot \Prob{X_j = 2} \right)
\\ & + 
\Prob{\Event[0,1]{E}} \cdot
\Prob{X_j \in \SET{0,1}}
\\ & + 
\Prob{\Event[1,2]{E}} \cdot
\tfrac12 \cdot
\Prob{X_j \in \SET{1,2}}, 
\end{align*}

\begin{align*}
\Prob{A' \text{ is EF1}}
\geq & \phantom{+}
\Prob{\Event[1]{E}} \cdot 
\Big( 
\Prob{X'_i = 0} \cdot \Prob{X'_j \in \SET{0,1}}
+ \Prob{X'_i = 1} \cdot \Prob{X'_j \in \SET{0,1,2}} 
\\ & \phantom{\Prob{\Event[1]{E}} \cdot } + \Prob{X'_i = 2} \cdot \Prob{X'_j \in \SET{1,2}} \Big)
\\ & + 
\Prob{\Event[0,1]{E}} \cdot
\Prob{X'_i \in \SET{0,1}} \cdot
\Prob{X'_j \in \SET{0,1}}
\\ & + 
\Prob{\Event[1,2]{E}} \cdot
\Prob{X'_i \in \SET{1,2}} \cdot
\Prob{X'_j \in \SET{1,2}}
\\ = & \phantom{+}
\Prob{\Event[1]{E}} \cdot 
\left( 
\tfrac34 \cdot \Prob{X'_j = 0} + \Prob{X'_j = 1} + \tfrac34 \cdot \Prob{X'_j = 2} \right)
\\ & + 
\Prob{\Event[0,1]{E}} \cdot
\tfrac34 \cdot \Prob{X'_j \in \SET{0,1}}
\\ & + 
\Prob{\Event[1,2]{E}} \cdot
\tfrac34 \cdot \Prob{X'_j \in \SET{1,2}}.
\end{align*}

Now, we recall that $\Prob{X_j = 0} = 2^{-s_j}, \Prob{X_j = 1} = s_j \cdot 2^{-s_j}$, and $\Prob{X_j=2} = \binom{s_j}{2} \cdot 2^{-s_j}$, and similarly for $X'_j$ with $s_j-1$ in place of $s_j$.
Then, the bounds above become:

\begin{align*}
\Prob{A \text{ is EF1}}
= & \phantom{+} \Prob{\Event[1]{E}} \cdot 
2^{-s_j} \cdot 
\left( 1 + s_j + \tfrac12 \cdot \binom{s_j}{2} \right)
\\ & + 
\Prob{\Event[0,1]{E}} \cdot 2^{-s_j} \cdot (s_j + 1)
\\ & + 
\Prob{\Event[1,2]{E}} \cdot
2^{-s_j} \cdot \tfrac12 \cdot \left( s_j + \binom{s_j}{2} \right),
\\
\Prob{A' \text{ is EF1}}
\geq & \phantom{+}
\Prob{\Event[1]{E}} \cdot 
\left( 
2^{-s_j} \cdot \tfrac32
+ 2^{-s_j} \cdot 2 \cdot (s_j-1)
+ 2^{-s_j} \cdot \tfrac32 \cdot \binom{s_j-1}{2} \right)
\\ & + 
\Prob{\Event[0,1]{E}} \cdot
2^{-s_j} \cdot \tfrac32 \cdot s_j
\\ & + 
\Prob{\Event[1,2]{E}} \cdot
2^{-s_j} \cdot \tfrac32 \cdot \left( s_j - 1 + \binom{s_j-1}{2} \right).
\end{align*}

We will now compare the multiplying factors in front of $\Prob{\Event[1]{E}}, \Prob{\Event[0,1]{E}}, \Prob{\Event[1,2]{E}}$; ignoring the common factor of $2^{-s_j}$, it suffices to show that for all $s = s_j \geq 3$, 
\begin{align}
1 + s + \tfrac12 \cdot \binom{s}{2} 
& \leq 
\tfrac32
+ 2 \cdot (s-1)
+ \tfrac32 \cdot \binom{s-1}{2}
\label{eqn:E1-bound}
\\ s + 1 & \leq \tfrac32 \cdot s
\label{eqn:E01-bound}
\\ s + \binom{s}{2}
& \leq 
3 \cdot \left( s - 1 + \binom{s-1}{2} \right).
\label{eqn:E12-bound}
\end{align}

Inequality~\eqref{eqn:E01-bound} is obvious. 
By multiplying Inequality~\eqref{eqn:E1-bound} by 2 and canceling out immediately common terms, it is equivalent to
\begin{align}
3 + \binom{s}{2} 
\leq 2s + 3 \cdot \binom{s-1}{2}.
\label{eqn:common-inequality}
\end{align}
Canceling common terms in Inequality~\eqref{eqn:E12-bound}, we see that it too is equivalent to Inequality~\eqref{eqn:common-inequality}.
The latter can be verified by substituting the definitions of Binomial coefficients and canceling out terms, showing that it is equivalent to $2s^2 \geq s+2$, which holds for all $s \geq 2$.

Combining the two arguments eliminating first all empty bundles, then unnecessary singleton bundles, we conclude that any allocation containing an empty bundle, or containing more than $\max(0,2n-m)$ singleton bundles, is dominated in EF1 probability by another allocation with no empty bundles and fewer singleton bundles.
This proves the lemma.
\end{proof}

\subsection{Algorithms for a Constant Number of Agents}
\label{subsec:weak-constant}

Theorem~\ref{thm:weak-exp-iid} shows that under the Membership Model --- unless the number of items is very small --- the number of iterations required to guarantee an EF1 allocation is exponential in the number of agents $n$.
Here, we show that an exponential dependence is indeed required \emph{only} for $n$; when $n$ is treated as a constant, then even under the Membership Model, an EF1 allocation can be found in time polynomial in the number of items $m$.
This holds very generally, namely, under general monotone valuation functions.

\begin{theorem}[Algorithm for a constant number of agents]
\label{thm:weak-constant-formal}
In the Membership Model for goods with monotone valuation functions, an EF1 allocation can be found using $O(m^{2n} \cdot (n!)^2)$ interactions and time.
When $n$ is a constant, this bound is polynomial in $m$.
\end{theorem}

\begin{proof}
The idea of the proof is straightforward: Our characterization result of Theorem~\ref{thm:exist-almost-contiguous} guarantees the existence of a structured EF1 allocation.
Furthermore, there are only polynomially many (in $m$) candidate allocations of this structure. 
An algorithm can exhaustively try all of them, and is guaranteed to find one that is EF1.

To make this idea precise, fix an arbitrary ordering of the items.
By Theorem~\ref{thm:exist-almost-contiguous},
there exists an EF1 allocation $A=(A_1, \ldots, A_n)$ such that each $A_i$ is the union of an interval and a single additional item.
Such an allocation is fully characterized by the $n-1$ boundaries between intervals, the choice of (at most) $n$ additional items, the matching between additional items and intervals, and the assignment of the resulting bundles to the agents.
There are at most $m^n$ choices each for the first two, and at most $n!$ choices for the latter two, resulting in an upper bound of $m^{2n} \cdot (n!)^2$ on the total number of allocations $A$ of the given type.

An algorithm can simply enumerate all of these allocations and propose them, and by Theorem~\ref{thm:exist-almost-contiguous} is guaranteed that at least one of the proposed allocations will be EF1.
Because the computational task of enumerating all of these allocations also takes constant (amortized) time per allocation, the running time is also $O(m^{2n} \cdot (n!)^2)$, proving the theorem. 
\end{proof}

The bound given in Theorem~\ref{thm:weak-constant-formal} was optimized for simplicity instead of tightness. 
For example, from the proof of Theorem~\ref{thm:exist-almost-contiguous}, we know that the added items must be at interval boundaries, reducing the search space.
Furthermore, the search space can be reduced if additional assumptions are made on the agents' valuation functions.
In particular, when the agents have \emph{identical} monotone valuations, instead of using Theorem~\ref{thm:exist-almost-contiguous}, we can use Lemma~5.3 of \citet{oh2021fairly}, which guarantees the existence of an EF1 allocation all of whose bundles are intervals.
Enumerating all cut positions induces all contiguous allocations, and since the valuations are identical, the bundle-agent assignments are immaterial.
Thus, proposing one allocation per cut position suffices, yielding an $O(m^n)$ interaction and time bound.

\section{The All-Witness EF1 Model: Algorithms via Envy-Cycle Elimination}
\label{sec:all-witness}
We now turn to the most informative feedback model for EF1, which we call the \emph{All-Witness EF1 Model}.
In this model, upon proposing an allocation $A$, the algorithm learns the complete set of agent pairs $(i,j)$ for which agent $i$ envies agent $j$’s bundle even after the removal of any single item, i.e., all pairs $i,j$ with $A_i \prec_i^{\mathrm{EF1}} A_j$.

Such rich feedback substantially simplifies the problem.
While the classical Envy-Cycle Elimination Algorithm cannot be applied directly, we show that it can be simulated with only minor modifications in the All-Witness Model.
This yields a polynomial-interaction and polynomial-time procedure for finding an EF1 allocation under arbitrary monotone valuations.

\subsection{Algorithmic Overview}
\label{subsec:all-witness-overview}
We begin by briefly recalling the classical Envy-Cycle Elimination Algorithm \citep{Lipton2004}, which incrementally constructs an EF1 allocation by maintaining a partial allocation and repeatedly eliminating cycles in the induced envy graph.
While this approach is well-suited to the full-information setting, it does not immediately apply to our interactive model with all-witness EF1 feedback.

There are two main obstacles.
First, envy-cycle elimination operates on partial allocations, whereas in our model, the algorithm must propose a complete allocation in every round.
Second, the algorithm relies on access to the full envy graph, i.e., access to every envy-free violation,  while the feedback we obtain consists only of EF1 violations, which is strictly weaker information.

We address these two issues separately.
To overcome the first obstacle, our algorithm designates the three agents $1,2,3$ to serve as \buffers which temporarily hold items outside the target partial allocation. 
For each partial allocation proposed by the Envy-Cycle Elimination Algorithm, our algorithm proposes three allocations.
In the first proposal, all such remaining items are assigned to agent~$1$ (in addition to items which agent 1 might be allocated in the partial allocation); this reveals all EF1 violations between pairs $(i,j)$ with $i,j \neq 1$.
In the second proposal, the remaining items are instead assigned to agent~$2$, which reveals all violations between pairs of agents $i,j \neq 2$.
Finally, assigning the remaining items to agent~$3$ reveals all EF1 violations between agents $i,j \neq 3$.
Together, these three proposals allow the algorithm to recover all EF1 violations among the agents under the target partial allocation.
This construction requires the presence of at least three agents.
When there are only $n=2$ agents, we have shown in Section~\ref{subsec:weak-constant} how to obtain an algorithm with polynomial complexity even in the Membership Model. In Appendix~\ref{sec:two-agent-case}, we show how to straightforwardly adapt ideas of \citet[Theorem~3.1]{oh2021fairly} to improve the number of iterations, and find an EF1 allocation in $O(\log m)$ iterations.

To overcome the second obstacle, we show that EF1 feedback suffices to infer a (possibly incomplete) directed envy graph over the agents in the partial allocation.
We then prove that applying envy-cycle elimination to this partial graph and assigning the next item to an agent with in-degree zero in the resulting graph preserves the EF1 property of the partial allocation.
Repeating this process allows the algorithm to simulate envy-cycle elimination until a complete EF1 allocation is obtained.

\subsection{The Algorithm for Three or More Agents}
\label{subsec:all-witness-main}

We now present the details of the algorithm for the All-Witness EF1 Model in instances with $n \ge 3$ agents and monotone valuation functions.

The algorithm proceeds in rounds, allocating one item per round, in the spirit of the classical envy-cycle elimination procedure.
Throughout the process, it maintains the invariant that the current partial allocation is EF1.

We fix an arbitrary ordering of the items.
Let $\tilde A^{(t)} = (\tilde A^{(t)}_1,\ldots,\tilde A^{(t)}_n)$ denote the partial allocation
maintained by the algorithm after round $t$, where $\tilde A^{(t)}_i \subseteq [m]$ is the bundle currently assigned to agent $i$.
Initially, no items are allocated, and the empty allocation $\tilde A^{(0)}$ trivially satisfies EF1.
Suppose that after some round the algorithm has constructed an EF1 partial allocation $\tilde A^{(t-1)}$, and consider the next item $t$ to be allocated.

The algorithm tentatively assigns this item to agents one by one.
For each tentative assignment, it uses all-witness EF1 feedback together with the buffer-agent construction to test whether the resulting partial allocation remains EF1.
If assigning the item to some agent preserves EF1, the algorithm fixes this assignment and proceeds to the next round.

If none of the tentative assignments preserves EF1, the algorithm exploits the information revealed by these failed attempts.
Specifically, the EF1 violations induced by adding the new item allow the algorithm to infer a set of directed envy relations that is consistent with the envy relations under the previous-round allocation $\tilde A^{(t-1)}$.
We write $E_{\mathrm{envy}}$ for the resulting set of directed pairs and $G_{\mathrm{envy}} = (N, E_{\mathrm{envy}})$ for the induced directed graph.
Intuitively, since $\tilde A^{(t-1)}$ was EF1 and all tentative assignments fail after the addition of item~$t$, each agent must be involved in some observed EF1 violation, implying that every agent has at least one incoming envy relation in $G_{\mathrm{envy}}$ and thus that $G_{\mathrm{envy}}$ contains a directed cycle.

When $G_{\mathrm{envy}}$ contains a directed cycle
$i_1 \to i_2 \to \cdots \to i_k \to i_1$, the algorithm performs a single cycle-rotation step:
each agent $i_\ell$ on the cycle receives the bundle
$\tilde A^{(t-1)}_{i_{\ell+1}}$, where indices are taken modulo $k$.
All bundles outside the cycle remain unchanged.
This operation updates the partial allocation while preserving the multiset of bundles, and enables the algorithm to restore progress when no tentative assignment of the current item preserves EF1.

Using this procedure, the algorithm identifies an agent to whom the current item can be assigned without violating EF1, restores the invariant, and advances to the next round.
Repeating this process for all items yields a complete EF1 allocation.

The algorithm is formally given as Algorithm~\ref{alg:all-witness-main}.

\begin{algorithm}[t]
\caption{All-Witness EF1 Allocation for $n \ge 3$ Agents}
\label{alg:all-witness-main}
\DontPrintSemicolon
\KwIn{Agents $[n]$, items $[m]$ in an arbitrary fixed order.}
\KwOut{An EF1 allocation with respect to the (unknown) true valuations.}

Initialize $\tilde A^{(0)}_i \gets \emptyset$ for all $i \in [n]$.\;

\For{$t = 1$ \KwTo $m$}{
    $\mathrm{assigned} \gets \textsc{false}$.\;

    \While{$\mathrm{assigned} = \textsc{false}$}{
        $E_{\mathrm{envy}} \gets \emptyset$.\;

        \ForEach{$i \in [n]$}{
            $E_i \gets \emptyset$.\;

            \tcp{Tentatively assign item $t$ to agent $i$}
            Define $\tilde A^{(t,i)}$ by
            \[
                \tilde A^{(t,i)}_i = \tilde A^{(t-1)}_i \cup \SET{t},
                \qquad
                \tilde A^{(t,i)}_j = \tilde A^{(t-1)}_j \text{ for all } j \neq i.
            \]

            \tcp{Probe EF1 violations using buffer agents}
            \ForEach{$b \in \SET{1,2,3}$}{
                Propose the complete allocation obtained by assigning all remaining items
                $\SET{t+1,\ldots,m}$ to agent $b$ and keeping $\tilde A^{(t,i)}$ unchanged otherwise.\;

                Receive all-witness EF1 feedback and update
                \[
                    E_i \gets E_i \cup
                    \Set{(p,q)}{(p,q)\text{ is reported and } p \neq b,\ q \neq b}.
                \]
            }

            \If{$E_i = \emptyset$}{
                \tcp{Item $t$ can be assigned to agent $i$ while preserving EF1}
                $\tilde A^{(t)} \gets \tilde A^{(t,i)}$.\;
                $\mathrm{assigned} \gets \textsc{true}$.\;
                \textbf{break}\;
            }
            \Else{
                $E_{\mathrm{envy}} \gets E_{\mathrm{envy}} \cup E_i$.\;
            }
        }

        \If{$\mathrm{assigned} = \textsc{false}$}{
            \tcp{All tentative assignments failed: rotate bundles along a directed cycle}
            Construct $G_{\mathrm{envy}} = (N, E_{\mathrm{envy}})$.\;

            Find any directed cycle
            $i_1 \to i_2 \to \cdots \to i_k \to i_1$
            in $G_{\mathrm{envy}}$.\;

            \For{$\ell = 1$ \KwTo $k$}{
                $\tilde A^{(t-1)}_{i_\ell} \gets \tilde A^{(t-1)}_{i_{(\ell+1) \mod k}}$.\;
            }
        }
    }
}

\Return{$\tilde A^{(m)}$}\;
\end{algorithm}

\begin{theorem}[All-witness EF1 for $n \ge 3$ agents]
\label{thm:all-witness-main}
In the All-Witness EF1 Model for goods with monotone valuation functions and $n \ge 3$ agents, Algorithm~\ref{alg:all-witness-main} outputs an EF1 allocation using $O(mn^3)$ interactions and $O(mn^5)$ time. 
\end{theorem}

\begin{proof}
By construction, the algorithm maintains the invariant that at the beginning of each round $t$ the current partial allocation $\tilde A^{(t-1)}$ satisfies EF1.
Therefore, once round $t$ terminates with $\tilde A^{(t)}$ obtained by fixing item $t$ to some agent, the resulting partial allocation remains EF1.
It thus suffices to prove that, for each $t \in [m]$, the inner \textbf{while} loop terminates after at most $O(n^2)$ iterations.
The claimed interaction and running-time bounds then follow by the accounting in Algorithm~\ref{alg:all-witness-main}.

Fix a round $t$ and consider an iteration of the \textbf{while} loop in which no tentative assignment succeeds.
In this case, the algorithm constructs a directed graph $G_{\mathrm{envy}} = (N, E_{\mathrm{envy}})$ from the observed all-witness EF1 feedback.
Since every tentative assignment fails, it follows that for every agent $i \in N$ the corresponding edge set $E_i$ is nonempty.
Because $i$ is the only agent other than $b$ who receives additional items compared to the allocation $\tilde{A}^{(t-1)}$ (which was EF1), the only possible EF1 violations must involve an agent envying $i$; that is, all edges in $E_i$ are directed edges towards $i$.
Thus, the fact that all $E_i \neq \emptyset$ implies that every vertex in $G_{\mathrm{envy}}$ has indegree at least~$1$.
Hence $G_{\mathrm{envy}}$ contains a directed cycle. 

We next relate $G_{\mathrm{envy}}$ to the envy relations under the previous-round allocation $\tilde A^{(t-1)}$.
By the construction of $E_{\mathrm{envy}}$ from the failed tentative assignments, each directed edge $(p,q) \in E_{\mathrm{envy}}$ indicates that, under the tentative allocation in which item $t$ is added, agent $p$ EF1-envies agent $q$.
Since item $t$ is the only change relative to $\tilde A^{(t-1)}$, such an EF1 violation can only arise if agent $p$ already envies agent $q$ under the previous-round allocation $\tilde A^{(t-1)}$; otherwise, the removal of item $t$ would eliminate $p$'s envy of $q$.
Consequently, $G_{\mathrm{envy}}$ is a subgraph of the envy graph induced by $\tilde A^{(t-1)}$.

The algorithm then selects an arbitrary directed cycle in $G_{\mathrm{envy}}$ and performs a single cycle-rotation step, reassigning each agent on the cycle to the bundle of their successor.
By the argument in \citet[Lemma~2.2]{Lipton2004} (which applies to any directed cycle of envy edges), this rotation strictly decreases the number of envy edges among the agents in $G_{\mathrm{envy}}$.
In particular, each time the algorithm performs a cycle-rotation step, the size of the edge set $E_{\mathrm{envy}}$ decreases by at least one.

Since $|E_{\mathrm{envy}}| \le n(n-1)$, after at most $n(n-1) = O(n^2)$ such iterations, $G_{\mathrm{envy}}$ contains a vertex/agent of in-degree zero.
For this agent, assigning item $t$ to $i$ preserves EF1, and thus the algorithm succeeds in the tentative assignment step and terminates round~$t$.

Therefore, the \textbf{while} loop terminates after at most $O(n^2)$ iterations in every round $t$.
As each \textbf{while} iteration uses at most $3n$ allocation proposals, the total number of proposals is $O(mn^3)$.
Moreover, processing each proposal takes $O(n^2)$ time, yielding a total running time of $O(mn^5)$.
This completes the proof.
\end{proof}

\section{Conclusion and Open Problems}
\label{sec:conclusion}
We defined an interactive model in which an algorithm must repeatedly propose allocations, and receives feedback either indicating that the proposed allocation is fair, or providing a specific fairness violation.
We have shown that under this model, for additive valuations, an algorithm can find an EF1 or PROP1 allocation in polynomial time.
When the valuations are merely monotone, an algorithm can still find an EF1 allocation in a polynomial number of interactions, but the internal computations are currently not polynomial-time.

\dkedit{Our framework is flexible and, as discussed in \cref{sec:generalizing-ellipsoid}, provides an interactive algorithm with a polynomial number of iterations whenever the notion of envy and valuations combined satisfy certain properties.
Another example of an application is for non-additive valuations with both goods and chores, and a slightly relaxed notion of envy in which an agent $i$ must not envy any agent $j$ after removal of at most one chore from $i$'s bundle \emph{and} one good from $j$'s bundle.
Under some restrictions on the class of valuations, 
\citet{Bilo2026approximately} generalized the existence result of \citet{Igarashi2023Discrete} to show existence of interval allocations satisfying this notion of envy-freeness; consequently, our algorithm and analysis from \cref{sec:monotone-ef1} can be directly applied, yielding an interactive algorithm with polynomially many interactions (but, to the best of our knowledge, not polynomial computation in each iteration).}

\dkreplace{This naturally}{Most importantly, our work} raises the question whether an EF1 almost-contiguous allocation --- whose existence we showed non-constructively --- can be found efficiently given explicit access to each agent's valuation of each almost-interval bundle.
The question parallels \dkreplace{a question of \citet{Bilo2022Connect}}{corresponding questions of \citet{Bilo2022Connect} and \citet{Igarashi2023Discrete}}: they proved non-constructively that an interval EF2 \dkedit{\citep{Bilo2022Connect} and EF1 \citep{Igarashi2023Discrete}}allocation always exists, but finding such an allocation in polynomial time \dkreplace{is}{was} also left as an open question.
\dkdelete{Another closely related question was also left open by \citet{Bilo2022Connect}: whether there always exists an interval \emph{EF1} allocation for monotone valuations.
The existence of such a valuation would improve the (already polynomial) number of interactions required by our algorithm for monotone valuations, by further narrowing down the search space.}

Relatedly, while in Appendix~\ref{sec:direct-almost-contiguous}, we showed NP-hardness of deciding whether a given blacklist of tuples still allows for an almost-contiguous allocation, the problem shown NP-hard is a slight generalization of the one faced by an algorithm trying to obtain an EF1 allocation.
Specifically, when the tuple $(i,B_1,B_2)$ is added to the blacklist, signifying that agent $i$, being given bundle $B_1$, would EF1-envy any agent given bundle $B_2$, this also implies that the same would be true for any tuple $(i,B'_1,B'_2)$ with $B'_1 \subseteq B_1, B'_2 \supseteq B_2$; the instances constructed in our reduction can violate this closure property, so it is conceivable that a more narrow definition of the problem that needs to be solved is not NP-hard.
After all, the \emph{existence} of an EF1 almost-contiguous allocation is \emph{guaranteed}, so the answer to the decision question for all instances created over the execution of a search algorithms is ``Yes''.  
Nonetheless, we view our NP-hardness result as strong evidence that the problem of finding an allocation avoiding all tuples of the blacklist may be computationally hard for a suitable complexity class of search problems.

Another related question is whether there is an alternative approach for interactively finding an EF1 allocation under general monotone valuation functions with single-witness feedback.
While our approach of embedding the objects linearly leads to a polynomial number of interactions, it certainly need not be the unique approach achieving this goal, and perhaps an alternative would lead to less challenging computational problems.

\dkdelete{While we showed the existence of almost-contiguous allocations for monotone allocations of \emph{goods}, our result does not extend to chores.
To the best of our knowledge, the existence of such allocations is open, as is the existence of interval EF2 allocations of chores.
The immediate adaptations of the Sperner-based proof of \citet{Bilo2022Connect} do not appear to work, though it is conceivable that a Sperner-based argument with more fundamental modifications could succeed.}

More fundamentally, other questions in the general realm of fairness and social choice could benefit from being studied through the lens of interactive learning. Doing so would pose a natural, and in many cases practically relevant, alternative to either assuming full access to inputs or specific types of interactions.

\dkedit{
\subsubsection*{Acknowledgments}
We would like to thank anonymous reviewers, conference attendees at EC 2026, \dkedit{Uriel Feige,} and Ayumi Igarashi for useful feedback and pointers.
}


\bibliographystyle{ACM-Reference-Format}
\bibliography{davids-bibliography/names, davids-bibliography/conferences, davids-bibliography/bibliography, davids-bibliography/publications, bibliography}

\clearpage
\appendix

\section{Naive Local-Adjustment Algorithms under the Single-Witness EF1 Model}
A natural question is whether EF1 under the Single-Witness Model can be achieved by a simple local-adjustment rule that directly responds to revealed envy.
Specifically, upon observing a violating pair $(i,j)$, one may transfer a single item from agent $j$ to agent $i$ and repeat this procedure until no further violations are reported.

We begin by analyzing this approach in the pure-goods setting with identical additive valuations.
In this case, we show that the local-adjustment algorithm terminates and outputs an EF1 allocation.
However, we can only establish an exponential upper bound on the number of iterations required by the local-adjustment process.
Furthermore, when valuations are not identical, the same procedure may fail to terminate.
The latter result implies that purely local adjustments are insufficient for obtaining efficient algorithms under the Single-Witness Model.

\begin{algorithm}[htb]
\caption{\textsc{LocalAdjustment-EF1} (Identical Additive Goods)}
\label{alg:greedy-transfer}
\DontPrintSemicolon
\SetKwInput{KwInput}{Input}
\SetKwInput{KwOutput}{Output}
\KwInput{An instance $(N,M)$ with identical additive valuations over pure goods.}
\KwOutput{An EF1 allocation with respect to the (unknown) true valuations.}
Initialize an arbitrary allocation $A^{(0)} = (A_1^{(0)}, \dots, A_n^{(0)})$.\;

\For{$t=0,1,2,\ldots$}{
    Propose the allocation $A^{(t)}$.\;
    \eIf{$A^{(t)}$ is certified EF1}{
        \Return $A^{(t)}$.\;
    }{
        Receive a violating pair $(i,j)$.\;
        Obtain the allocation $A^{(t+1)}$ by transfering an arbitrary item $x \in A_j^{(t)}$ to agent $i$, and leaving all other bundles unchanged.
    }
}
\end{algorithm}

\begin{lemma}\label{lem:local-adjustment-identical}
Under identical additive valuations in the pure-goods setting, Algorithm~\ref{alg:greedy-transfer} outputs an EF1 allocation using at most \(O(n^{m})\) interactions.
Each interaction requires only constant-time computation once a violating pair is revealed.
\end{lemma}

\begin{proof}
Let \(A^{(t)} = (A_1^{(t)}, \dots, A_n^{(t)})\) denote the allocation at iteration \(t\),
and define the potential function
\[
\Phi(A^{(t)}) = \sum_{i \in N} v(A_i^{(t)})^2.
\]

Consider an iteration \(t\) in which the current allocation \(A^{(t)}\) is not EF1, and a violating pair \((i,j)\) is revealed.
The algorithm transfers an arbitrary item \(x \in A_j^{(t)}\) to agent \(i\).
Let \(v(x)\) denote the value of item \(x\), and define
\[
A_i^{(t+1)} = A_i^{(t)} \cup \{x\}, \qquad
A_j^{(t+1)} = A_j^{(t)} \setminus \{x\}.
\]
The EF1 violation guarantees that 
\(\left(v(A_j^{(t)}) - v(x)) - v(A_i^{(t)}\right) > 0\).
We thus have
\begin{align*}
\Phi(A^{(t+1)}) - \Phi(A^{(t)})
&= (v(A_i^{(t)}) + v(x))^2 + (v(A_j^{(t)}) - v(x))^2
   - v(A_i^{(t)})^2 - v(A_j^{(t)})^2 \\
&= -2v(x)\big(v(A_j^{(t)}) - v(A_i^{(t)}) - v(x)\big) \le 0.
\end{align*}
Equality holds only when \(v(x) = 0\). Hence, the potential function never increases throughout the algorithm.

\begin{description}
\item[Case 1: \(v(x) > 0\).]
In this case, \(\Phi(A^{(t+1)}) < \Phi(A^{(t)})\).
Since the potential strictly decreases in every such step,
the allocation \(A^{(t+1)}\) cannot have appeared in any earlier iteration \(t' \le t\),
as all previous allocations satisfy \(\Phi(A^{(t')}) \ge \Phi(A^{(t)}) > \Phi(A^{(t+1)})\).

\item[Case 2: \(v(x) = 0\).]
In this case, the potential remains unchanged,
\(\Phi(A^{(t+1)}) = \Phi(A^{(t)})\).
We claim that \(A^{(t+1)}\) has not appeared before.
Suppose, for contradiction, that the same allocation first appeared at iteration \(t_1 < t_2 = t+1\).
Since \(\Phi(A^{(t_1)}) = \Phi(A^{(t_2)})\) and the potential never increases,
all transfers between times \(t_1\) and \(t_2\) must have involved items with zero value,
so the total value of each agent’s bundle remained constant.
For the same allocation to reappear, the item \(x\) must have been moved away from
agent \(i\) after \(t_1\) and eventually returned from agent \(j\) at \(t_2\).
Let $t_1 \leq t'_1 < t'_2 < \cdots < t'_k = t_2$ be the time steps at which item $x$ was transferred, with the transfer happening at time $t'_{\ell}$ going from $i_{\ell}$ to $i_{\ell+1}$.
As a result, we have that $i_1 = i, i_k = i$, and $v(A_{i_{\ell+1}}^{t'_{\ell}}) < v(A_{i_{\ell}}^{t'_{\ell}})$ for all $\ell=1, \ldots, k-1$.
Because all transferred items have value 0, we also have that $v(A_{i_{\ell+1}}^{t'_{\ell}}) = v(A_{i_{\ell+1}}^{t'_{\ell+1}})$.
Together, these imply that 
\[
v(A_i^{t_2}) = v(A_{i_k}^{t'_k}) < v(A_{i_{k-1}}^{t'_{k-1}}) < \cdots < v(A_{i_1}^{t'_1}) = v(A_i^{t'_1}).
\]
However, this is a contradiction to the fact that only value-0 items have been transferred between times $t_1$ and $t_2$.
Hence, \(A^{(t+1)}\) cannot have appeared before.
\end{description}

Since allocations never repeat,
and there are at most \(n^m\) distinct allocations,
Algorithm~\ref{alg:greedy-transfer} must terminate after at most \(n^m\) iterations.
The final allocation must be EF1 because this is explicitly checked as a termination condition.

Finally, the algorithm performs at most \(n^m\) iterations,
each involving a single item transfer and constant local computation once the violating pair \((i,j)\) is revealed.
Hence, the total interaction complexity is \(O(n^m)\),
and the per-iteration computational complexity is \(O(1)\).
\end{proof}

The convergence argument above critically relies on the valuations being identical.
When valuations differ across agents, the local adjustment rule of Algorithm~\ref{alg:greedy-transfer} no longer guarantees progress, as illustrated by a simple counterexample.

Consider two agents \(1,2\) and three items \(x_1,x_2,x_3\) with additive valuations:
\[
\begin{array}{c|ccc}
 & x_1 & x_2 & x_3 \\
 \hline
 v_1 & 1 & 5 & 2 \\
 v_2 & 5 & 1 & 2 \\
\end{array}
\]
Suppose the current allocation is \(A_1=\SET{x_1,x_3}\) and \(A_2=\SET{x_2}\).
In this allocation, agent \(2\) envies \(1\) by more than one item.
Transferring \(x_3\) from \(1\) to \(2\) reverses the envy, making \(1\) envy \(2\);
moving \(x_3\) back restores the original allocation, creating a cycle.
Thus, under non-identical valuations, the process may oscillate indefinitely.

Intuitively, when agents evaluate bundles differently, both may regard the other’s bundle as more desirable, so single-item transfers cannot eliminate mutual envy.
One might attempt a bundle-swap variant that performs a full swap between \(A_i\) and \(A_j\) once a previously visited allocation reoccurs.
However, even this modification can enter cycles.
Tables~\ref{tab:valuation} and~\ref{tab:trace} present a concrete instance\footnote{This instance was obtained by generating random additive valuations and running the bundle-swap variant of the local-adjustment dynamics.
It is included solely as a concrete counterexample, and the details of the generating process are therefore omitted.} with non-identical additive valuations, involving 6 agents and 12 items, whose execution under the bundle-swap variant of the local-adjustment dynamics forms a cycle.

This example highlights a clear distinction between identical and heterogeneous additive valuations.
Under identical valuations, Algorithm~\ref{alg:greedy-transfer} is guaranteed to terminate, since each one-item transfer decreases a natural potential function.
When valuations differ across agents, this structural monotonicity is lost.
Even allowing full bundle swaps is insufficient to prevent cycling, indicating that heterogeneous additive valuations admit fundamentally more complex local-improvement behavior than their identical counterparts.
This separation helps explain why purely local adjustment rules are insufficient in general, and motivates the use of more structured methods that ensure global progress with polynomial interaction, as developed in the main body of the paper.

\begin{table}[htb]
\centering
\small
\setlength{\tabcolsep}{5pt}
\caption{Valuation matrix for a concrete instance with 6 agents and 12 items.}
\label{tab:valuation}
\begin{tabular}{c *{12}{c}}
\toprule
Agents/Items 
& 1 & 2 & 3 & 4 & 5 & 6
& 7 & 8 & 9 & 10 & 11 & 12 \\
\midrule
Agent 1 & 21 & 45 & 14 & 97 & 45 & 98 & 80 & 88 & 57 & 86 & 24 & 29 \\
Agent 2 & 38 & 63 & 51 & 54 & 38 & 59 & 49 & 42 & 49 & 32 & 50 & 58 \\
Agent 3 & 51 & 40 & 48 & 35 & 68 & 38 & 36 & 34 & 52 & 37 & 97 & 57 \\
Agent 4 & 22 & 45 & 12 & 95 & 40 & 89 & 80 & 83 & 47 & 88 & 17 & 36 \\
Agent 5 & 24 & 59 & 28 & 72 & 71 & 25 & 33 & 39 & 31 & 27 & 71 & 67 \\
Agent 6 & 24 & 49 & 9  & 97 & 46 & 90 & 69 & 81 & 52 & 89 & 14 & 30 \\
\bottomrule
\end{tabular}
\end{table}

\begin{table}[htb]
\centering
\small
\setlength{\tabcolsep}{4pt}
\caption{Execution trace of the local-adjustment dynamics on the instance in Table~\ref{tab:valuation}.}
\label{tab:trace}
\begin{tabular}{l l l c c c c c c}
\toprule
Step & Trigger & Action
& Agent 1 & Agent 2 & Agent 3 & Agent 4 & Agent 5 & Agent 6 \\
\midrule
0  & -- & Initial
& 1,2 & 3,4 & 5,6 & 7,8 & 9,10 & 11,12 \\

1  & violate pair $(1,4)$ & $4 \to 1$: Item 8
& 1,2,8 & 3,4 & 5,6 & 7 & 9,10 & 11,12 \\

2  & violate pair $(5,1)$ & $1 \to 5$: Item 2
& 1,8 & 3,4 & 5,6 & 7 & 2,9,10 & 11,12 \\

3  & violate pair $(4,5)$ & $5 \to 4$: Item 10
& 1,8 & 3,4 & 5,6 & 7,10 & 2,9 & 11,12 \\

4  & violate pair $(6,3)$ & $3 \to 6$: Item 6
& 1,8 & 3,4 & 5 & 7,10 & 2,9 & 6,11,12 \\

5  & violate pair $(2,6)$ & $6 \to 2$: Item 6
& 1,8 & 3,4,6 & 5 & 7,10 & 2,9 & 11,12 \\

6  & violate pair $(1,2)$ & $2 \to 1$: Item 6
& 1,6,8 & 3,4 & 5 & 7,10 & 2,9 & 11,12 \\

7  & violate pair $(3,1)$ & $1 \to 3$: Item 1
& 6,8 & 3,4 & 1,5 & 7,10 & 2,9 & 11,12 \\

8  & violate pair $(6,1)$ & $1 \to 6$: Item 6
& 8 & 3,4 & 1,5 & 7,10 & 2,9 & 6,11,12 \\

9  & violate pair $(2,6)$ & $6 \to 2$: Item 6
& 8 & 3,4,6 & 1,5 & 7,10 & 2,9 & 11,12 \\

10 & violate pair $(1,2)$ & $2 \to 1$: Item 6
& 6,8 & 3,4 & 1,5 & 7,10 & 2,9 & 11,12 \\

11 & repeat step 7 & Swap$(1,2)$
& 3,4 & 6,8 & 1,5 & 7,10 & 2,9 & 11,12 \\

12 & violate pair $(6,2)$ & $2 \to 6$: Item 6
& 3,4 & 8 & 1,5 & 7,10 & 2,9 & 6,11,12 \\

13 & violate pair $(2,1)$ & $1 \to 2$: Item 4
& 3 & 4,8 & 1,5 & 7,10 & 2,9 & 6,11,12 \\

14 & violate pair $(1,2)$ & $2 \to 1$: Item 4
& 3,4 & 8 & 1,5 & 7,10 & 2,9 & 6,11,12 \\

15 & repeat step 12 & Swap$(1,2)$
& 8 & 3,4 & 1,5 & 7,10 & 2,9 & 6,11,12 \\

16 & repeat step 8 & Swap$(1,2)$
& 3,4 & 8 & 1,5 & 7,10 & 2,9 & 6,11,12 \\

17 & repeat step 12 & Cycle detected
& 3,4 & 8 & 1,5 & 7,10 & 2,9 & 6,11,12 \\
\bottomrule
\end{tabular}
\end{table}
\label{sec:local_adjustment}

\section{Another Instantiation of the Ellipsoid Framework: EFX with Charity}
We present another instantiation of the ellipsoid-based framework for a relaxation of envy-freeness, namely \emph{EFX-with-charity}.

We begin by briefly recalling the notion of \emph{envy-freeness up to any good (EFX)}.
An allocation is EFX if no agent envies another agent after the removal of \emph{any} single good from the other agent’s bundle (rather than a carefully chosen single good).
EFX is strictly stronger than EF1 and is regarded as one of the most compelling relaxations of envy-freeness for indivisible goods.
However, unlike EF1, it is not known whether EFX allocations always exist, even for relatively simple valuation classes.
To address this, \citet{Caragiannis2019efxc} introduced a relaxation that allows a small number of goods to remain unallocated, or equivalently, donated to charity.
This leads to the notion of \emph{EFX-with-charity}.

\begin{definition}[EFX-with-charity for goods]
We say that agent $i$ \emph{EFX-envies} agent $j$, and write $A_i \prec_i^{\mathrm{EFX}} A_j$, if there exists a good $x \in A_j$ such that
\[
  v_i(A_i) < v_i(A_j \setminus \{x\}).
\]
A (partial) allocation $(A,C)$ is \emph{EFX-with-charity} if there is no pair $i,j$ with
$A_i \prec_i^{\mathrm{EFX}} A_j$, where $C$ denotes the set of unallocated goods and $A=(A_1,\dots,A_n)$ is a partition of $M \setminus C$.
\end{definition}

\citet{Caragiannis2019efxc} showed that under additive valuations, there exists an EFX-with-charity allocation that achieves a Nash social welfare at least half of the optimal Nash social welfare achievable on the full set of goods~$M$.
Subsequently, \citet{Chaudhury2021efxc} strengthened this result by introducing the notion of \emph{EFX-with-bounded-charity}.
In this variant, the goal is to compute EFX-with-charity allocations that satisfy explicit bounds on the charity set~$C$, including both a bound on its cardinality and a bound on its value.
In particular, they showed that for general monotone valuations, there always exists an EFX-with-charity allocation satisfying $|C| < n$ and $v_i(C) \le v_i(A_i)$ for all agents~$i$.
Their proof is constructive and yields a pseudo-polynomial-time algorithm for computing such allocations.

Having defined the notion of EFX with charity, we now describe the interactive feedback model for EFX, analogous to the single-witness models introduced in Section~\ref{subsec:models} for EF1 and PROP1.

Compared to EF1 and PROP1, the EFX condition is more demanding in that a violation depends not only on a pair of agents, but also on the specific good whose removal fails to eliminate envy.
Therefore, identifying an EFX violation requires specifying a concrete witness good $x \in A_j$, which is essential for deriving a valid separating hyperplane in the ellipsoid-based framework.

\begin{definition}[Single-Witness EFX-with-charity Model]
In the Single-Witness EFX-with-charity Model, upon proposing a (partial) allocation $(A,C)$, the feedback either \emph{certifies} that $(A,C)$ is EFX-with-charity, or consists of an \emph{adversarially chosen violating triple} $(i,j,x) \in N \times N \times M$ with
\[
  v_i(A_i) < v_i(A_j \setminus \{x\}).
\]
We call the triple $(i,j,x)$ a \emph{witness} of EFX-with-charity violation.
\end{definition}

Under this feedback model, the resulting Algorithm~\ref{alg:ellipsoid-efx-charity} differs from Algorithm~\ref{alg:ellipsoid-ef1} only in the type of witness feedback and the corresponding separating constraints.
Since we focus on the pure-goods setting, each agent's valuation vector lies in \([0,1]^m\) rather than \([-1,1]^m\).
In addition, the allocation step now invokes the pseudo-polynomial-time EFX-with-bounded-charity algorithm of~\citet{Chaudhury2021efxc} as a black box.
Correctness and termination follow by the same ellipsoid-based argument as for EF1 and PROP1.

\begin{algorithm}[t]
\caption{\textsc{Ellipsoid-EFX-with-charity} (Single-Witness Feedback, Goods, Additive Valuations)}
\label{alg:ellipsoid-efx-charity}
\DontPrintSemicolon
\SetKwInput{KwInput}{Input}
\SetKwInput{KwOutput}{Output}
\KwInput{Agents $N$ and goods $M$.}
\KwOutput{An EFX-with-charity (partial) allocation with respect to the (unknown) true additive valuations.}
\BlankLine

\lForEach{$i\!\in\!N$}{
  Initialize (normalized) feasible region $\mathcal{R}_i^{(0)}\!\gets\![0,1]^m$.}
$t \gets 0$.\;
\While{no proposed allocation has been certified as EFX-with-charity}{
  $t \gets t+1$.\;

  \tcp{\textbf{(a) Valuation update and allocation step}}
  $\vval[k][(t)] \gets \mathrm{Center}\big(\mathcal{R}_k^{(t-1)}\big)$ for all $k \in N$.\;

  $(A^{(t)},C^{(t)}) \gets \textsc{BoundedCharityEFX}(\vval[][(t)])$.\;
  Propose the (partial) allocation $(A^{(t)},C^{(t)})$.\;
  \lIf{$(A^{(t)},C^{(t)})$ is certified EFX-with-charity}{\Return $(A^{(t)},C^{(t)})$.}

  \BlankLine
  \tcp{\textbf{(b) Identify new separating constraints}}
  Receive a violating triple $(i,j,x)$.\;

  Define the constraint
  $  \mathcal{C}_i^{(t)} :=
    \Set{\val[i]}{v_i(A_i^{(t)}) \le v_i\big(A_j^{(t)} \setminus \SET{x}\big)}.
  $ 

  \BlankLine
  \tcp{\textbf{(c) Ellipsoid update}}
  $\mathcal{R}_i^{(t)} \gets \mathcal{R}_i^{(t-1)} \cap \mathcal{C}_i^{(t)}$.\;
  $\mathcal{R}_k^{(t)} \gets \mathcal{R}_k^{(t-1)}$ for all $k \neq i$.\;
}
\end{algorithm}

\begin{theorem}
\label{thm:efx-charity-ellipsoid}
Under additive valuations in the goods setting, Algorithm~\ref{alg:ellipsoid-efx-charity} terminates after $\poly(n,m)$ rounds of proposed allocations.
Its total running time is pseudo-polynomial, dominated by the calls to the EFX-with-bounded-charity subroutine of~\citet{Chaudhury2021efxc}.
\end{theorem}

\begin{proof}[Proof sketch]
The proof mirrors the arguments for EF1 and PROP1.
In each round, if the proposed allocation is not certified, the feedback provides a violating triple $(i,j,x)$, which identifies a corresponding violated EFX inequality for agent~$i$.
Intersecting the current feasible region with the resulting separating constraint eliminates the current center while preserving feasibility of the true valuation.
Therefore, the ellipsoid update strictly shrinks the feasible region whenever the algorithm does not terminate.

By standard complexity guarantees of the ellipsoid method for additive valuation vectors in dimension~$m$, the algorithm terminates after $\poly(n,m)$ rounds.
The total running time is pseudo-polynomial due to the bounded-charity EFX subroutine.
\end{proof}

\label{sec:efx-charity}

\section{Direct Search over Almost-Contiguous Allocations}
In Section~\ref{subsec:polyquery-almost-contiguous}, we showed that by restricting attention to almost-contiguous allocations, each strong EF1 violation rules out a distinct local configuration, which we record as a forbidden tuple in a blacklist $\mathcal{B}$.
This observation yields a polynomial bound on the number of EF1 interactions, but does not by itself lead to a polynomial-time algorithm due to the search procedure of knife positions.

We now describe an alternative, purely combinatorial approach that avoids the ellipsoid framework, and with it the search for knife positions.
The algorithm directly maintains the blacklist $\mathcal{B}$ of forbidden tuples and iteratively proposes almost-contiguous allocations that avoid all tuples blacklisted so far.
The algorithm fixes an arbitrary ordering of the items as $1,2,\dots,m$ and considers only allocations that respect this order.

In each round $t$, the algorithm (Algorithm~\ref{alg:elimination-basic-boundary}) proposes an allocation $A^{(t)}$.
If the allocation is certified as EF1, the algorithm terminates.
Otherwise, the feedback returns a violating ordered pair $(i,j)$, which induces a new forbidden tuple that is added to $\mathcal{B}$.
The process continues until an EF1 allocation is found or all feasible tuples have been excluded.

\begin{algorithm}[htb]
\caption{\textsc{EliminationSearchAlmostContiguous} (Single Witness Feedback, EF1, Monotone Goods)}
\label{alg:elimination-basic-boundary}
\DontPrintSemicolon
\SetKwInput{KwInput}{Input}
\SetKwInput{KwOutput}{Output}

\KwInput{Instance $(N,M)$ under monotone valuations with goods ordered as $1,\dots,m$.}
\KwOutput{An EF1 allocation with respect to the (unknown) true monotone valuations.}

Initialize the blacklist $\mathcal{B} \gets \emptyset$.\;

\Repeat{$A$ is EF1.}{
    Partition the items $1,\dots,m$ into $n$ bundles $A=(A_1, \ldots, A_n)$
    such that:\;
    \Indp
        Each $A_i$ is the union of an interval plus at most one additional item.\;
        For each blacklisted tuple $(i,B,B') \in \mathcal{B}$, either $A_i \neq B$, or $A_j \neq B'$ for all $j \neq i$.\;
    \Indm

    Propose the allocation $A=(A_1,\dots,A_n)$.\;
    \lIf{a violating pair $(i,j)$ is returned with $A_i \prec_i^{\mathrm{EF1}} A_j$}{
    update the blacklist
    $\mathcal{B} \gets \mathcal{B} \cup \SET{(i,A_i,A_j)}$.}
}
\end{algorithm}

The correctness and interaction complexity of Algorithm~\ref{alg:elimination-basic-boundary} follow exactly the same reasoning as in Theorem~\ref{thm:ellipsoid-polyquery-monotone} for the ellipsoid-based algorithm.
In particular, each violation feedback rules out a distinct forbidden tuple, and the total number of iterations is therefore bounded by the number of such tuples.
This is captured by the following theorem:

\begin{theorem}[Polynomial Interaction Bound] \label{thm:polyquery-monotone}
Under monotone valuations in the pure-goods setting,
Algorithm~\ref{alg:elimination-basic-boundary} terminates after at most \(O(n m^6)\) proposals and outputs an almost-contiguous EF1 allocation.
\end{theorem}

While Algorithm~\ref{alg:elimination-basic-boundary} avoids the search over knife positions implicit in the algorithm of Section~\ref{sec:monotone-ef1} due to its non-constructive existence proof of an EF1 allocation, it replaces the problem with another difficult problem: finding an almost-contiguous allocation avoiding all blacklisted tuples.
In fact, we show that a (slightly generalized) decision version of this new problem is NP-hard.

\begin{theorem}[NP-Completeness of Allocation Step]
\label{thm:NP-completeness}
Deciding if there exists an almost-contiguous allocation that is consistent with a given blacklist $\mathcal{B}$ is NP-complete. 
The problem remains NP-complete even under the restriction that every bundle occurring in a tuple of $\mathcal{B}$ is an interval and the sought allocation is contiguous.
\end{theorem}

\begin{proof}
We define the decision version \textsc{Forbidden-Tuple Elimination} as follows: given $n,m$, and a blacklist $\mathcal{B}$ of forbidden tuples $(i,B,B')$ (where $B \cap B' = \emptyset$ for each tuple), determine whether there exists an almost-contiguous allocation $A=(A_1,\dots,A_n)$ such that for all tuples $(i,B,B') \in \mathcal{B}$, either $A_i \neq B$ or $A_j \neq B'$ for all $j \neq i$.

Membership in NP is obvious, with an allocation $A$ serving as a witness.
We prove NP-completeness via a polynomial-time reduction from \textsc{Independent Set}.
Given an instance of \textsc{Independent Set} consisting of a graph $G=(V,E)$
and a target size $k$, we construct an instance of
\textsc{Forbidden-Tuple Elimination}.

\newcommand{\RItem}[2]{a_{#1,#2}}

\begin{enumerate}
    \item For each vertex $v_i\in V$, create two consecutive items $\RItem{i}{1}, \RItem{i}{2}$.
    The total number of items is $m=2|V|$, ordered as
    \[
        \RItem{1}{1}, \RItem{1}{2}, \RItem{2}{1}, \RItem{2}{2}, \dots, \RItem{|V|}{1}, \RItem{|V|}{2}.
    \]
    \item Set the number of agents to $n = 2|V| - k$.
    \item Construct the blacklist $\mathcal{B}$ as follows.
    \begin{enumerate}
        \item Include all tuples $(i,B,B')$ in which $B$ is not an interval, or $B'$ is not an interval.
        \item For each edge $(v_i,v_j)\in E$, add the tuple $(\ell,[\RItem{i}{1},\RItem{i}{2}],[\RItem{j}{1},\RItem{j}{2}])$ to $\mathcal{B}$ for every agent $\ell$.
        \item For any pair of consecutive items $\RItem{i}{2},\RItem{i+1}{1}$ that do not belong to the same vertex, 
        add the tuples $(\ell,[\RItem{i}{2},\RItem{i+1}{1}], B)$ and $(\ell,B,[\RItem{i}{2},\RItem{i+1}{1}])$ to $\mathcal{B}$, for all intervals $B$ and agents $\ell$.
        \item Finally, forbid all intervals of length three or greater by adding to $\mathcal{B}$ all tuples $(\ell, B, B')$ such that at least one of $B, B'$ has size at least three.
    \end{enumerate}
\end{enumerate}

\medskip
\noindent
The number of tuples added above is polynomial in $m$ and $n$, and all such tuples can be constructed in polynomial time.
Therefore, the reduction runs in polynomial time.
To prove correctness of the reduction, we first observe that it constructs a blacklist that is symmetric with respect to the agents, i.e., whenever $(\ell, B, B') \in \mathcal{B}$, so is $(\ell', B, B')$ for all $\ell'$.
Next, by the definition of the blacklist $\mathcal{B}$, an allocation $A=(A_1, \ldots, A_n)$ is compatible with $\mathcal{B}$ if and only if the following all hold:
\begin{enumerate}
    \item Each bundle $A_i$ is an interval of at most two items.
    \item If $A_{\ell}$ contains two items (as opposed to zero or one), these two items correspond to the same vertex $v_i \in V$.
    \item If $A_{\ell}, A_{\ell'}$ both contain two items, the vertices $v_i, v_{i'}$ corresponding to the items are not adjacent.
\end{enumerate} 

With this characterization, the formal correctness proof is straightforward.
If $S \subseteq V$ is an independent set of size at least $k$, then for each $v_i \in S$, assign the bundle $[\RItem{i}{1}, \RItem{i}{2}]$ to an agent, and assign the remaining $2|V|-2k$ items as singletons to the remaining $2|V|-2k$ agents. 
By the preceding characterization, this allocation does not violate any tuple in $\mathcal{B}$.

Conversely, if there is an allocation not violating any tuple in $\mathcal{B}$, then at least $k$ agents must receive bundles of size two, which must be of the form $[\RItem{i}{1}, \RItem{i}{2}]$. 
Let $S$ be the set of all vertices $v_i \in V$ such that an agent receives the corresponding bundle of size two.
By the characterization, $S$ must be independent, and of size at least $k$.

\smallskip
Therefore, the \textsc{Forbidden-Tuple Elimination} problem is NP-complete.
\end{proof}

\label{sec:direct-almost-contiguous}

\section{Proof of Lemma~\ref{lem:binom-two-point}}
\begin{proof}[Proof of Lemma~\ref{lem:binom-two-point}]
Write $p_t := \Prob{X=t} = \binom{s}{t} \cdot 2^{-s}$.
Since $p_t$ is symmetric and unimodal (increasing up to the mode(s) and then decreasing),
the quantity $p_k+p_{k+1}$ is maximized when $\SET{k,k+1}$ is centered at the mode(s).
Thus the maximum over $k$ equals
\[
M_s :=
\max_k (p_k+p_{k+1})
=
\begin{cases}
\left( \binom{2r}{r-1}+\binom{2r}{r} \right) \cdot 2^{-2r}, & s=2r,\\[1.0ex]
\left( \binom{2r+1}{r}+\binom{2r+1}{r+1} \right) \cdot 2^{-(2r+1)}, & s=2r+1.
\end{cases}
\]

For the even case $s=2r$, note that $\binom{2r}{r-1}=\binom{2r}{r}\cdot \frac{r}{r+1}$, so
\[
M_{2r}
=
\binom{2r}{r}\left(1+\frac{r}{r+1}\right) \cdot 2^{-2r}
=
\binom{2r}{r} \cdot 2^{-2r} \cdot \frac{2r+1}{r+1}.
\]
A direct ratio computation yields
\[
\frac{M_{2r+2}}{M_{2r}}
=
\frac{\binom{2r+2}{r+1} \cdot \frac{2r+3}{r+2} \cdot 2^{-(2r+2)}}{\binom{2r}{r} \cdot \frac{2r+1}{r+1}\cdot 2^{-2r}}
=
\frac14\cdot
\frac{\binom{2r+2}{r+1}}{\binom{2r}{r}}
\cdot
\frac{\frac{2r+3}{r+2}}{\frac{2r+1}{r+1}}
=
\frac{2r+3}{2(r+2)}
<1,
\]
so $M_{2r}$ is strictly decreasing in $r$.
Since $M_2=3/4$, we have $M_{2r}\le 3/4$ for all $r\ge 1$.

For the odd case $s=2r+1$, using $\binom{2r+1}{r}=\binom{2r+1}{r+1}$ we get
\[
M_{2r+1}= 2 \cdot \binom{2r+1}{r} \cdot 2^{-(2r+1)} = \binom{2r+1}{r} \cdot 2^{-2r}.
\]
Substituting this representation into a ratio bound,
\[
\frac{M_{2r+3}}{M_{2r+1}}
=
\frac{\binom{2r+3}{r+1} \cdot 2^{-(2r+2)}}{\binom{2r+1}{r} \cdot 2^{-2r}}
=
\frac{1}{4}\cdot \frac{\binom{2r+3}{r+1}}{\binom{2r+1}{r}}
=
\frac{(2r+3)(2r+2)}{4(r+1)(r+2)}
=
\frac{2r+3}{2(r+2)}
\le 1.
\]
Hence $M_{2r+1}$ is nonincreasing in $r$, and $M_3=3/4$ implies $M_{2r+1}\le 3/4$ for all $r\ge 1$.

Combining the even and odd cases proves $\Pr[X\in\SET{k,k+1}] \le M_s \le 3/4$ for all $s\ge 2$.
\end{proof}
\label{sec:binom-two-point}

\section{Achieving $O(\log m)$ interactions for Two Agents under All-Witness EF1 Feedback}
We saw that for any constant number of agents, EF1 allocations can be found in polynomial time even under the Membership EF1 Model.

Here, we show that under the All-Witness Model, for $n=2$ agents, it is possible to find an EF1 allocation using only $O(\log m)$ rounds of interaction.
Our approach is based on \citet[Theorem~3.1]{oh2021fairly}, who showed that, given access to \emph{EF} feedback, their Algorithm~1 finds a contiguous-interval EF1 allocation using $O(\log m)$ queries.

In our setting, the feedback consists only of all-witness EF1 violation information.
Nevertheless, we show that the same binary-search idea can be implemented using only EF1 feedback, yielding an $O(\log m)$ interaction and time bound.

\begin{algorithm}[htb]
\caption{(for two agents with monotone valuations under all-witness EF1 feedback)}
\label{alg:two-agent-all-witness}
\textbf{Step 1:}
Arrange the items on a line in an arbitrary order.

For each cut position $k \in \SET{0,\ldots,m}$, define a prefix
$P_k := \SET{1,\ldots,k}$ and a suffix
$S_k := \SET{k+1,\ldots,m}$ (with $P_0=\emptyset$ and $S_m=\emptyset$).

Using binary search, find the rightmost cut position $k \in \SET{0,\ldots,m}$ such that for the allocation $(P_k,S_k)$ we have $P_k \prec_1^{\mathrm{EF1}} S_k$.
Terminate early if an EF1 allocation is found before determining~$k$.

\medskip
\noindent
\textbf{Step 2:}
If $S_k \prec_2^{\mathrm{EF1}} P_k$, propose $(S_k, P_k)$,
else propose $(S_{k+1}, P_{k+1})$.
\end{algorithm}

\begin{theorem}[Algorithm for two agents]
\label{thm:two-agent-all-witness}
In the All-Witness EF1 Model for goods with monotone valuation functions and two agents,
Algorithm~\ref{alg:two-agent-all-witness} finds an EF1 allocation using $O(\log m)$ interactions and time.
\end{theorem}

\begin{proof}
We first argue that the algorithm terminates using $O(\log m)$ interactions and time.
Consider Step~1.
If there is no cut position $k$ such that proposing $(P_k,S_k)$ yields $P_k \prec_1^{\mathrm{EF1}} S_k$, then in particular agent~1 does not EF1-envy agent~2
under $(P_0,S_0)=(\emptyset,M)$.
Since agent~2 receives the grand bundle in this allocation, they also do not EF1-envy agent~1.
Hence $(\emptyset,M)$ is EF1, and the binary search terminates at $k=0$.

Otherwise, such a cut position exists.
Binary search therefore finds such a $k$ in $O(\log m)$ iterations.
Moreover, we must have $k < m$, since under $(P_m,S_m)=(M,\emptyset)$ agent~1 cannot EF1-envy agent~2.
Thus, in Step~2, the algorithm proposes exactly one allocation, either at cut $k$ or at cut $k+1$,
and both indices are well-defined.
Overall, the number of proposed allocations is $O(\log m)$, and the running time is $O(\log m)$.

We next prove that the algorithm outputs an EF1 allocation.
If no such $k$ exists, then as argued above the algorithm reaches $k=0$, yielding an EF1 allocation.
Now suppose such a $k$ exists.
If the algorithm terminates early in Step~1, then by construction it has already proposed an EF1 allocation.
Otherwise, the algorithm enters Step~2.
By the definition of the rightmost cut $k$, we have
\[
P_k \prec_1^{\mathrm{EF1}} S_k
\quad\text{and}\quad
P_{k+1} \not\prec_1^{\mathrm{EF1}} S_{k+1}.
\]
Moreover, entering Step~2 implies that proposing $(P_{k+1},S_{k+1})$ was not EF1, and therefore we must have
\[
S_{k+1} \prec_2^{\mathrm{EF1}} P_{k+1}.
\]

We consider the two cases in Step~2.

\smallskip
\noindent
\textbf{Case 1:} $S_k \prec_2^{\mathrm{EF1}} P_k$.
Then under the allocation $(P_k,S_k)$ the two agents EF1-envy each other.
After swapping the bundles and proposing $(S_k,P_k)$, neither agent EF1-envies the other, and the resulting allocation is EF1.

\smallskip
\noindent
\textbf{Case 2:} $S_k \not\prec_2^{\mathrm{EF1}} P_k$.
In this case the algorithm proposes $(S_{k+1},P_{k+1})$.
This swap ensures that agent~2 does not EF1-envy agent 1, since we have $S_{k+1} \prec_2^{\mathrm{EF1}} P_{k+1}$ for the unswapped allocation.
It remains to argue that agent~1 does not EF1-envy agent~2 under $(S_{k+1},P_{k+1})$.
Using the fact that $P_k \prec_1^{\mathrm{EF1}} S_k$ and the relation between the two cuts, we have
\[
v_1\!\left(P_{k+1}\setminus\SET{k+1}\right)
= v_1(P_k)
< v_1\!\left(S_k\setminus\SET{k+1}\right)
= v_1(S_{k+1}),
\]
where the strict inequality follows from the EF1 violation $P_k \prec_1^{\mathrm{EF1}} S_k$.
Therefore, after the swap, agent~1 also does not EF1-envy agent~2, and the allocation is EF1.

In both cases, Step~2 outputs an EF1 allocation. This completes the proof.
\end{proof}
\label{sec:two-agent-case}

\end{document}